\documentclass[letterpaper,twocolumn,10pt]{article}
\usepackage{usenix}
\usepackage{amsmath,amsthm,amsfonts,amssymb,amscd}
\usepackage{lastpage}
\usepackage{enumerate}
\usepackage{mathrsfs}
\usepackage[x11names]{xcolor}
\usepackage{graphicx}
\usepackage{listings}
\usepackage{hyperref}
\usepackage{url}
\usepackage{bm}
\usepackage{xspace}
\usepackage{framed}
\usepackage[camera]{dtrt}
\usepackage{thmtools}
\usepackage[capitalize]{cleveref}
\usepackage{enumitem}
\usepackage{multicol}
\usepackage{tablefootnote}
\usepackage{subcaption}

\usepackage{mdframed}
\usepackage{booktabs}
\usepackage{mathtools}
\usepackage{tikz}
\usetikzlibrary{arrows.meta,positioning,calc,patterns.meta}%
\usepackage{pgfplots}\pgfplotsset{compat=1.18}\usepgfplotslibrary{groupplots}%
\usepackage[most]{tcolorbox}
\usepackage{algorithm}    %
\usepackage{algorithmicx}
\usepackage{algpseudocode}%
\usepackage{siunitx}
\usepackage{acro}
\usepackage{multirow}
\newtcolorbox{gamebox}[1]{
    enhanced,
    breakable,
    colback=white,
    colframe=black,
    coltitle=white,
    colbacktitle=black,
    title={#1},
    fonttitle=\bfseries,
    boxrule=0.65pt,
    arc=0pt,
    left=5pt,
    right=5pt,
    top=5pt,
    bottom=5pt,
    before skip=6pt,
    after skip=6pt,
    attach boxed title to top left={xshift=0pt,yshift=0pt},
    boxed title style={
        colback=black,
        colframe=black,
        arc=0pt,
        boxrule=0pt,
        left=5pt,
        right=5pt,
        top=3pt,
        bottom=3pt
    }
}

\newcommand{\F}{\mathbb{F}}

\renewcommand{\S}{\mathbb{S}}
\makeatletter
\def\defcal#1{\expandafter\def\csname cal#1\endcsname{\mathcal{#1}}}
\@for\c:=A,B,C,D,E,F,G,H,I,J,K,L,M,N,O,P,Q,R,S,T,U,V,W,X,Y,Z\do{\expandafter\defcal\expandafter{\c}}

\def\defbf#1{%
  \@ifundefined{b#1}{%
    \expandafter\def\csname b#1\endcsname{\mathbf{#1}}%
  }{}%
}
\@for\c:=A,B,C,D,E,F,G,H,I,J,K,L,M,N,O,P,Q,R,S,T,U,V,W,X,Y,Z,a,b,c,d,e,f,g,h,i,j,k,l,m,n,o,p,q,r,s,t,u,v,w,x,y,z\do{\expandafter\defbf\expandafter{\c}}
\makeatother

\newcommand{\wt}{\operatorname{wt}}
\newcommand{\supp}{\operatorname{supp}}
\newcommand{\bias}{\operatorname{bias}}
\newcommand{\pfail}{p_{\mathrm{fail}}}
\DeclareAcronym{mpc}{
    short = MPC,
    long = multi-party computation
}
\DeclareAcronym{fhe}{
    short = FHE,
    long = fully homomorphic encryption
}
\DeclareAcronym{vMVMD}{
    short = vMVMD,
    long = verifiable matrix--vector multiplication delegation
}
\DeclareAcronym{pvMVMD}{
    short = pvMVMD,
    long = private and verifiable matrix--vector multiplication delegation
}
\DeclareAcronym{tee}{
    short = TEE,
    long = trusted execution environment
}
\DeclareAcronym{llm}{
    short = LLM,
    long = large language model
}
\DeclareAcronym{mvm}{
    short = MVM,
    long = matrix--vector multiplication
}
\newcommand{\bPi}{\boldsymbol{\Pi}}
\newcommand{\PivMVMD}{\Pi_{\mathsf{vMVMD}}}
\newcommand{\PipvMVMD}{\Pi_{\mathsf{pvMVMD}}}

\newcommand{\bDelta}{\boldsymbol{\Delta}}
\newcommand{\pp}{\mathsf{pp}}
\newcommand{\adv}{\mathcal{A}}
\newcommand{\st}{\mathsf{st}}
\newcommand{\td}{\mathsf{td}}
\newcommand{\GRAA}{\mathbf{G}_{\text{RAA}}}

\newcommand{\negl}{\mathsf{negl}}

\newcommand{\privN}{N_{\text{x}}}
\newcommand{\verN}{N_{\text{y}}}
\newcommand{\privR}{R_{\text{x}}}
\newcommand{\verR}{R_{\text{y}}}
\newcommand{\privG}{\mathbf{G}_{\text{x}}}
\newcommand{\verG}{\mathbf{G}_{\text{y}}}
\newcommand{\privcalG}{\mathcal{G}_{\text{x}}}
\newcommand{\vercalG}{\mathcal{G}_{\text{y}}}
\newcommand{\privt}{t_{\text{x}}}
\newcommand{\verit}{t_{\text{y}}}
\newcommand{\privdelta}{\delta_{\text{x}}}
\newcommand{\verdelta}{\delta_{\text{y}}}
\newcommand{\prive}{\mathbf{e}_{\text{x}}}
\newcommand{\vere}{\mathbf{e}_{\text{y}}}
\newcommand{\verl}{\ell}

\newcommand{\protname}{\textsf{Maverick}\xspace}

\usepackage{float}
\floatstyle{plain}
\newfloat{protocol}{htbp}{prot}
\floatname{protocol}{Protocol}
\crefname{protocol}{Protocol}{Protocols}
\Crefname{protocol}{Protocol}{Protocols}

\usepackage{pifont}

\declaretheorem[numberwithin=section]{theorem}
\declaretheorem[sibling=theorem]{lemma}

\declaretheorem[sibling=theorem]{definition}
\declaretheorem[sibling=theorem]{remark}

\title{\protname: Private and Verifiable LLM Inference Made Practical via\\Matrix--Vector Multiplication Delegation}
\author{
{\rm Ben Merbaum$^*$}\\
Yale University
\and
{\rm Mohammad Amin Raeisi$^*$}\\
Yale University
\and
{\rm Wenhao Wang$^*$}\\
Yale University, IC3
\and
{\rm Charalampos Papamanthou}\\
Yale University
\and
{\rm Katerina Sotiraki}\\
Yale University
\and
{\rm Fan Zhang}\\
Yale University, IC3
}

\begin{document}
\maketitle
\def\thefootnote{*}\footnotetext{These authors contributed equally to this work and are ordered alphabetically.}\def\thefootnote{\arabic{footnote}}
\def\thefootnote{\dag}\footnotetext{Emails: \{\texttt{ben.merbaum}, \texttt{amin.raeisi},\\ \texttt{wenhao.wang}, \texttt{charalampos.papamanthou},\\ \texttt{katerina.sotiraki}, \texttt{f.zhang}\}@\texttt{yale.edu}.}\def\thefootnote{\arabic{footnote}}

\begin{abstract}
Open-source \acp{llm} are increasingly competitive with closed-source models while offering transparency and the ability to run inference without exposing user inputs to a service provider.
However, running large-scale models locally requires substantial computational resources.
In practice, users may still resort to a third-party provider, giving rise to privacy and correctness concerns.
Existing solutions that address these problems often impose substantial server overhead or introduce additional trust assumptions.

In this paper, we present \protname, a novel approach to private and verifiable \ac{llm} inference based on a protocol for delegating matrix--vector multiplication, a dominant operation in \acp{llm}.
At its core, \protname provides, to our knowledge, the first information-theoretically sound verification protocol for matrix--vector multiplication delegation with \textit{transparent preprocessing}, \textit{efficient (batch) verification}, and \textit{virtually no server overhead}. We combine this verification primitive with LPN-based pseudorandom masking to provide input privacy.

We implement \protname and evaluate both of its modes on large open-source models.
When only verification is required, \protname lets commodity clients use models far beyond their own capacity: a laptop and a desktop with a 16\,GB GPU generate tokens $81$--$270\times$ faster on \texttt{Qwen3-32B} than with local inference, and reach $7$--$33$ tokens per second on \texttt{Llama~4}, a 400B-parameter model that neither can run at all.
When input privacy is also required, \protname generates tokens $10$--$12\times$ faster on \texttt{Qwen2.5-72B} than the concurrent system MOSAIC, using precomputed masks.

\end{abstract}

\section{Introduction}\label{sec:intro}

Open-source models such as DeepSeek, Qwen, and Kimi are increasingly matching the performance of closed-source models~\cite{deepseekV3,qwen3technicalreport,kimiK2}.
For example, as of July 2026, Kimi K3 outperforms Claude models on frontend design tasks~\cite{designarenaKimiK3}.
In addition, open-source models can be run locally, allowing users with sufficient hardware to keep their inputs private from service providers.
Their publicly available weights also offer better transparency~\cite{fmti} and reproducibility~\cite{chen2023chatgpt}, while reducing dependence on provider-controlled changes to pricing, model availability, and access policies.
Despite various advantages, their substantial computational requirements make local deployment impractical or undesirable for most users (e.g., running Kimi K3 requires a cluster of high-end GPUs estimated to cost \$0.55--\$1 million to set up~\cite{duola2026kimi}).
In practice, users resort to specialized service providers, reintroducing privacy and correctness concerns. 

The problem of private and verifiable outsourcing of AI and \ac{llm} inference has been extensively studied.
However, aside from recent systems based on \acp{tee}~\cite{redpillConfidentialAI,anthropicConfidentialInference,cryptoAISurvey}, existing cryptographic approaches typically incur significant overhead on the server side.
For instance, proof systems~\cite{ghodsi2017safetynets,ZEN21,zkCNN21, verfCNN,zkLLM24,benno2025joltatlas,QuickSilver21, Mystique21,AntMan22,lycklama2025artemis,zkGPT25, deepprove} and homomorphic encryption~\cite{CryptoNets16,MiniONN17, Gazelle18,Delphi20,Cheetah22,Iron22,BOLT24,Rhombus24, BumbleBee25,Nexus25,BLB25} have been used, either individually or in combination, to provide verifiability, input privacy, or both. However, generating proofs and performing \ac{fhe} computations can be orders of magnitude more expensive than the original workload.
Solutions based on \acp{tee} offer practical performance, though their security guarantees may not always hold in practice~\cite{chuangTEEFailBreakingTrusted2026}.

In this paper, we present a lightweight protocol for private and verifiable outsourced \ac{llm} inference.
Our protocol allows a \textit{computationally constrained client}, such as a mobile device, to use an open-source model run by a powerful, untrusted server while hiding the client's private input and supporting verification, all with \textit{negligible} server overhead.

In particular, our technical contributions are as follows.
\textbf{1)} A protocol for fast \ac{vMVMD}. To our knowledge, it is the first information-theoretically sound construction with transparent preprocessing and efficient verification runtime. \textbf{2)} We combine \ac{vMVMD} with pseudorandom masks using coding theory to hide client inputs without introducing server overhead, providing \ac{pvMVMD}. \textbf{3)} Building on \ac{pvMVMD}, we present \protname, a system for private and verifiable \ac{llm} inference.
\protname outsources the linear operations of \ac{llm} inference to an untrusted server.
Nonlinear operations~\cite{vaswani2017attention} are inexpensive by comparison and are performed locally.
We implement \protname with a GPU server and evaluate it on open-source models of up to 400B parameters.

\subsection{Verifiable Matrix--Vector Multiplication Delegation (vMVMD)}
\label{subsec:intro-approach}

In a \ac{vMVMD} protocol, a client with input $\bx\in\F^n$ verifies that a server's output $\by\in\F^m$ satisfies $\by=\bM\bx$ for a public matrix $\bM\in\F^{m\times n}$.
Since the client can always compute $\by$ locally in $O(mn)$ time, verification time must be less than $O(mn)$ for delegation to be meaningful.
One idea is to apply Freivalds' algorithm~\cite{Freivalds1977ProbabilisticMC}. Given matrices $\bA,\bB,\bC$, Freivalds' algorithm checks whether $\bA\bB\stackrel{?}{=}\bC$ by sampling a random vector $\br$ and checking
$(\br^\top \bA)\bB \stackrel{?}{=} \br^\top \bC$.
However, when applied to the \ac{mvm} $\bM\bx\stackrel{?}{=}\by$, it requires the client to compute $\br^\top \bM$, which is as expensive as computing $\by$ locally.

Our idea is to use a \textit{sparse} vector in Freivalds' algorithm.
Rather than sampling a uniform challenge $\br$, we sample a sparse vector $\be$ from $\F^m$ with $t$ nonzero elements.
To check $\bM\bx\stackrel{?}{=}\by,$
the client computes $(\be^\top\bM)\bx$ and $\be^\top\by$ in time $O(tn)$ and accepts if and only if they match.
When $t=O(\kappa)$, verification takes $O_\kappa(m+n)$ time instead of $O(mn)$ time.

However, this strawman protocol is not sound. 
Consider an invalid claim $\bM\bx\neq\by$ for which $\bM\bx$ and $\by$ differ at only one index, say $k$. The above protocol only rejects if the $k$-th element of $\be$ happens to be nonzero, which happens with probability $t/m$, i.e., the soundness error is at least $1-t/m$.

To reduce the soundness error, we use \emph{error-correcting codes} to perform checks on encoded values.
Let $\bG\in\F^{m\times N}$ be the generator matrix of a linear code.
The client samples a sparse challenge (of length $N$) and checks $(\be^\top\bG^\top\bM)\bx=\be^\top\bG^\top\by$.
Intuitively, the encoding spreads every nonzero error vector $\bDelta\coloneqq\bM\bx-\by$ into a codeword $\bG^\top\bDelta\in\F^N$ with enough nonzero entries that its support is likely to intersect the support of $\be$.
This requires a code with high minimum Hamming weight (equivalently, high relative distance) and also fast encoding for fast client online time, properties satisfied by, for example, RAA codes~\cite{distanceRAAcodes}.

Computing $\bG^\top\bM$ is relatively heavy, but $\bQ=\bG^\top \bM$ is independent of client inputs and thus can be prepared in a \textit{preprocessing phase}.
The preprocessing adds no extra trust assumption because both $\bG$ and $\bM$ are public and anyone can verify the correctness of $\bQ$.
With $\bQ$ preprocessed, verification is efficient since $\be$ is sparse and $\bG$ has linear-time encoding. 

To further optimize client verification time, we develop a protocol for efficiently verifying a batch of \acp{mvm}
involving different matrices and vectors. For each matrix $\bM_b$, preprocessing produces $\bQ_b=\bG^\top\bM_b$. The same sparse challenge $\be$ can be shared across all claims in a repetition, but the client would still need to compute the expensive vectors $\bQ_b^\top\be$ separately for every matrix. We instead stack the computation of $\bQ_b^\top \be$ into a single \ac{mvm} and delegate it to the server. We then recursively apply \ac{vMVMD} to verify this computation. Crucially, its input is now the sparse vector $\be$ itself, so the second check can be evaluated efficiently by the client. 

\subsection{Private and Verifiable Matrix--Vector Multiplication Delegation (pvMVMD)}
We now show how to add privacy (i.e., hiding $\bx$ from the server) on top of vMVMD without introducing server overhead.
A naive construction is to have the client sample a uniformly random mask $\bx'$ and send $\widehat{\bx}=\bx+\bx'$ to the server; upon receiving $\widehat{\by}$, the client performs vMVMD to check $\bM\widehat{\bx}\stackrel{?}{=}\widehat{\by}$.
To remove the mask from $\widehat{\by}=\bM(\bx+\bx')$, however, the client needs to compute $\bM\bx'$ and subtract it from $\widehat{\by}$;
computing $\bM\bx'$ takes $O(mn)$ work.

To avoid this overhead, we use another code $\privG$ to generate pseudorandom privacy masks $\bx'=\privG\prive$  from sparse random vectors $\prive$ under standard cryptographic assumptions.
To distinguish this code from the one used for verification, we denote it by $\privG$, using the subscript $\text{x}$ because privacy masks are added to the input $\bx$.
The resulting mask can be removed efficiently as follows.
Because $\bM\privG$ is independent of the client's input, it can be computed during preprocessing.
With $\bP\coloneqq\bM\privG\in\F^{m\times\privN}$, we can write $\bM\bx'=(\bM\privG)\prive=\bP\prive$.
Since $\prive$ is sparse, computing $\bP\prive$ requires combining only the columns of $\bP$ corresponding to the nonzero entries of $\prive$.
In our construction, we use a code that satisfies the above requirement under the dual-LPN assumption~\cite{expand-accumulate-codes}. 

Overall, the client time is $O_{\lambda,\kappa}(m+n)$, while the server performs the ordinary multiplication without any overhead.

\subsection{\protname: Private and Verifiable Outsourced LLM Inference}
\label{subsec:intro-ml}
Building on \ac{pvMVMD}, we present \protname, a system for private and verifiable outsourced \ac{llm} inference.
An \ac{llm} produces its output by processing the prompt in a prefill phase and then generating tokens one by one in a decode phase.
Both phases involve computing a \textit{forward pass} for each token.
The key design of \protname is to delegate the linear operations of the forward passes (mainly \acp{mvm}), which account for the bulk of the computation, to the server, and to have the client compute the nonlinear operations locally, such as normalization, attention, and activation functions.
The design is simple, yet it relieves the client of most of the work of inference.
We highlight several design choices below.

\parhead{Real-time vs.\ optimistic verification}
The client can verify the server's results in real time, i.e., after each generated token, or in one batch at the end, a strategy we call \emph{optimistic verification}.
The former has the benefit of showing the user only verified tokens, while the latter yields a throughput gain due to batch verification.

\parhead{Verification-only vs.\ privacy}
\protname can operate in two modes.
The \emph{verification-only mode} provides integrity alone: the client learns that the response is the output of the claimed model, while the server sees the prompt, as it does in today's inference services.
This is a useful setting in its own right, and the target of recent systems such as Anthropic's FlashProve~\cite{flashprove}.
We optimize this mode by making it completely non-interactive: once the prompt is sent, the server evaluates the entire model and streams the results of the \acp{mvm} to the client, which replays the nonlinear operations and verifies the results.
A request thus takes a single round trip, so its performance does not depend on the latency of the network.

In the \emph{privacy mode}, the client additionally masks the inputs of the \acp{mvm} and unmasks their results, which hides the prompt and the response from the server.
Each \ac{mvm} then takes a round of interaction, so this mode requires a low-latency link between the client and the server.
Generating the masks is a significant part of the client's work, so the client can generate them ahead of time, similar to how Beaver triples are pre-generated in \ac{mpc} protocols~\cite{beavertriples}; we call this the \emph{boost mode}.

\parhead{Quantization}
\ac{pvMVMD} operates over a finite field, whereas a model computes in floating point.
\protname therefore quantizes the weights and the inputs of the \acp{mvm} into integers;
we use an adaptive quantization strategy designed to ensure that no result can wrap around the modulus.
Because the client returns to floating point between \acp{mvm}, inputs are quantized afresh and the integers do not grow from one \ac{mvm} to the next, allowing us to use a small field of 31 bits.
A small field has two benefits.
First, field arithmetic can be performed efficiently with the native integer operations of GPUs.
Second, it reduces the traffic between the client and the server (vectors of field elements).

\parhead{Other optimizations}
Our implementation includes further optimizations that reduce the number of rounds in the privacy mode, compress the server's outputs in the verification-only mode, and let a client with a small GPU evaluate part of the model itself to save communication.

\subsection{Implementation and Evaluation}

We implement our \ac{pvMVMD} protocol and \protname in about 38,000 lines of Rust and 2,000 lines of CUDA. 
We now highlight evaluation results.

\parhead{Delegation primitive}
We evaluate \ac{pvMVMD} on square matrices of dimension up to $n=2^{14}$.
The client time grows approximately linearly with $n$, while the preprocessing and native \ac{mvm} grow quadratically.
At $n=2^{14}$, the client completes its online work in \SI{7.27}{\milli\second}, whereas performing the native matrix--vector multiplication locally would take \SI{135.12}{\milli\second}, yielding a $18.60\times$ client-side speedup.
We compare the \ac{pvMVMD} protocol against Dumas--Zucca~\cite{DumasZucca17}, which is specialized for verifying matrix--vector products, and against Sum-Check instantiated with BaseFold~\cite{thaler13,basefold}, which represents the general proof-based approach used by recent verifiable-ML systems, including DeepProve~\cite{deepprove}.
Across the evaluated dimensions up to $n=2^{12}$, our implementation is up to $34.8\times$ faster than Dumas--Zucca and up to $194.8\times$ faster than Sum-Check+BaseFold, and our reported client time additionally includes input privacy.

\parhead{Verification-only mode}
The verification-only mode targets computationally weak clients that cannot run large models on their own.
We use two widely used open-source models: \texttt{Qwen3-32B}~\cite{qwen3-32b-fp8} and \texttt{Llama~4}~\cite{llama4-maverick-fp8}, a model of 400B parameters.
We evaluate it on two such clients: a laptop with an 8-core CPU, 30\,GB of memory and no discrete GPU, and a desktop with a 16\,GB consumer GPU and 16\,GB of memory.
The clients reside in a university network and the servers in data centers; they communicate over a WAN with a round-trip time of about 10\,ms and a bandwidth of 1\,Gb/s.

As the baseline, each client runs the model locally with llama.cpp~\cite{llamacpp}, a widely used engine for local inference, in its 8-bit format.
We report two quantities that a user experiences: the \emph{time to first token} (TTFT), from sending the prompt until the first token is available and verified, and the \emph{generation throughput}, the number of tokens generated per second.

For \texttt{Qwen3-32B}, which doesn't fit in the memory of either client, local inference is impractical: a 128-token response takes close to an hour on the laptop.
With \protname, the same clients achieve a generation throughput $81$--$270\times$ higher, and a TTFT about $2\times$ shorter on the laptop and $7$--$11\times$ shorter on the desktop, which processes the prompt on its own GPU.
For \texttt{Llama~4}, which neither client can run at all, the laptop and the desktop generate $7$--$15$ and $24$--$33$ tokens per second.
By contrast, FlashProve~\cite{flashprove}, the state-of-the-art zero-knowledge proof system for \ac{llm} inference, reports \SI{58.8}{\second} on an H100 GPU to prove the computation of a single next token of the same model after a 1,024-token prompt, and about \SI{21}{\second} to verify the proof.

\parhead{Privacy mode}
Because the privacy mode takes a round of interaction for each \ac{mvm}, it is practical where the client and the server are connected over a low-latency link.
Such a setting arises in data centers: a trusted GPU, e.g., one running in a \ac{tee}, seeks to increase its throughput by outsourcing the bulk of its computation to untrusted accelerators.
This is the setting of the concurrent system MOSAIC~\cite{chiang2026mosaic}
We compare \protname with MOSAIC on \texttt{Qwen2.5-72B}, with H200 GPUs connected by NVLink in the boost mode.

MOSAIC is the natural point of comparison: it targets the same deployment and, like \protname, outsources the linear operations of the model to untrusted GPUs.
The two systems make different trade-offs in what they guarantee: \protname verifies the server's results and keeps the outputs exact, whereas MOSAIC also hides the model from the server, but it does not verify the results and adds noise to the outputs.
The comparison therefore shows how efficiently each design delivers its guarantees in the same setting.
With one untrusted GPU instead of MOSAIC's three, \protname processes prompts $2.6$--$4.8\times$ faster and generates tokens $10$--$12\times$ faster, and its trusted side spends $19$--$27$\,ms per generated token against $43$\,ms.
These numbers exclude the preparation of \protname's masks, which the boost mode performs ahead of time, whereas MOSAIC generates its masks within the request.

\parhead{Accuracy}
Because \protname computes over a finite field, it must quantize the model, and we measure what this costs in accuracy.
We run each model on the WikiText-2 test set~\cite{merity2017pointer}, the standard corpus for evaluating quantized models, once with the weights released by its developers and once with \protname's quantization, and compare the perplexity of the two runs.
For the two Qwen models, the perplexity differs by at most $0.04\%$, and for \texttt{Llama~4} by $0.2\%$.
The quantization of \protname thus incurs only a minor loss in accuracy.

\section{Related Work}
\label{sec:related}

HE- and MPC-based private-inference systems protect client inputs and, in some settings, model parameters, but replace native linear algebra with encrypted or secure computation~\cite{CryptoNets16, SecureML17,MiniONN17,Gazelle18,Delphi20,cryptflow2,Cheetah22,Iron22,Rhombus24,BOLT24,Sigma24,Nexus25,BumbleBee25,Shaft25}. Specialized protocols exploit the structure of linear computation to
reduce this overhead. EMVP~\cite{EMVP25}, for instance, uses structured codes to efficiently delegate encrypted matrix--vector products while additionally hiding the stored matrix, but assumes a semi-honest server and cannot provide verifiability. Braverman and Newman~\cite{braverman2025practical} use trapdoored matrices for delegated linear algebra, including matrix multiplication and batched \acp{mvm} with efficiency in the amortized setting. Relatedly, Abbaszadeh et al. \cite{abbaszadeh2025single} use dual-LPN-friendly codes and transparent preprocessing to privately delegate multi-scalar multiplications. However, these works do not provide efficient verifiable delegation of a single \ac{mvm} in our setting.

A large body of work verifies neural-network and \ac{llm} inference using interactive proofs and SNARKs~\cite{ghodsi2017safetynets,ZEN21,zkCNN21,zkLLM24, zkGPT25,deepprove,verfCNN,lycklama2025artemis,ModularSumcheck23}. These systems can certify substantially more of the inference computation, but require additional proof-generation work at the server. Since linear layers account for a substantial fraction of inference computation, efficient verification of matrix products is a central subproblem. A direct Freivalds-style check provides information-theoretic
soundness, but the verification requires quadratic time~\cite{Freivalds1977ProbabilisticMC}. Prior work has also combined Freivalds-style verification with coding theory to reduce its randomness complexity \cite{bennett2024matrixmultiplicationverificationusing}; however, the resulting verification time remains quadratic. Slalom~\cite{slalom} performs quadratic verification work in a private preprocessing by computing a secret reusable Freivalds state; however, this state encounters leakage over a bounded number of queries. Further, Slalom uses fresh one-time masks to provide privacy for \acp{mvm}, but the corresponding demasking term is computed in preprocessing via an independent \ac{mvm}, leaving quadratic preprocessing work per query. Slalom at the Carnival~\cite{satc} reduces this per-query privacy-preprocessing cost by deriving fresh masks from reusable public structure. However, we show in \cref{app:satc} that its vector-space instantiation can leak linear information about the client's input, violating its claimed input-privacy guarantee. 

Publicly verifiable schemes such as Fiore--Gennaro~\cite{FG12}, Dumas--Zucca~\cite{DumasZucca17}, and general sum-check or SNARK-based approaches~\cite{thaler13} instead require authentication or proof generation in addition to the underlying computation, introducing considerable additional server-side work beyond the target \ac{mvm}. LAMP~\cite{LAMP2026} also combines Freivalds-style verification with error-correcting codes, but uses them inside a commit-and-prove SNARK for matrix multiplication.

Other systems that provide both privacy and verifiability include verifiable computation over encrypted data through FHE+ZKP and verifiable-FHE constructions~\cite{FGP14,vFHE23,HELIOPOL24,Laminate25, FioreVCAHE25, fiore2020, phalanx}, SNARGs specialized to matrix multiplication over encrypted data \cite{practicalSNARGMatmul26}, plaintext-authentication approaches such as VERITAS \cite{veritasfhe24}, whose main encodings were subsequently cryptanalyzed in the fully homomorphic setting\cite{VeritasAttack25}, VOLE-based proof systems ~\cite{QuickSilver21,Mystique21,AntMan22, Wolverine21}, DataSeal's FHE-based redundancy checks~\cite{dataseal}, and schemes combining homomorphic matrix computation with Freivalds-style verification ~\cite{zhang2020,liu2023}. These approaches require encrypted server computation, an auxiliary proof or authentication protocol, trusted hardware, or client-specific secret verification state. In contrast, Maverick achieves linear online client work with reusable transparent preprocessing and no secret verifier state; its base protocol requires neither trusted hardware nor heavy cryptographic machinery, and the server performs only one field \ac{mvm}.

\parhead{Concurrent work} We are aware of two concurrent papers, MOSAIC~\cite{chiang2026mosaic}, and Gupta et al.~\cite{gupta26}, which have similarly employed matrix--vector multiplication delegation toward outsourcing LLM inference.
MOSAIC provides privacy for the client input and model, but does not provide verifiability.
Like \protname, MOSAIC achieves asymptotically optimal online client time for an \ac{mvm}; however, under the concrete parameterizations of the two protocols, \protname has a substantially smaller constant, while its online client work additionally includes information-theoretic verification. Moreover, MOSAIC relies on both LWE and LPN and introduces Gaussian noise into the delegated computation, resulting in a nonzero protocol-error probability, whereas the MVM delegation of \protname introduces no such error. Furthermore, Gupta et al.~\cite{gupta26} provide both privacy and verifiability, but security relies on two non-colluding servers, whereas \protname outsources inference to only a single server. Meanwhile, their verifiability mechanism requires a trusted setup party (as opposed to our transparent preprocessing) not to collude with the server performing outsourced computation; the verification protocol requires group rather than field operations, and soundness is computational rather than information-theoretic.
Nevertheless, we note that both concurrent works require substantially less client-side storage for preprocessing state, which is a key bottleneck of our approach.

\section{Preliminaries}
\parhead{Notation}
We use $\F$ to denote a finite field, $\F^\ast=\F\setminus\{0\}$, and $\F_q$  for a prime power $q$ to denote the field with $q$ elements. 
Matrices are denoted by \textbf{bold} uppercase letters (e.g., $\bM$) and vectors by \textbf{bold} lowercase letters (e.g., $\bx$). For a vector $\bv$, we write $\wt(\bv)$ to denote its Hamming weight, the number of nonzero entries of $\bv.$ 
For a vector $\bv$ of length $n$, we use $\supp(\bv)$ to denote the support, the set of indices in $[n]$ where $\bv$ is nonzero. We say that a vector $\bv$ is $t$-sparse for a constant $t$ if $\wt(\bv)=t$. We denote the set of $t$-sparse vectors in $\F^n$ by $\calE_{t}^{n}$. 

We denote $\lambda$ as a computational security parameter and $\kappa$ as a statistical parameter. 
$\negl(\cdot)$ denotes a negligible function and $X\approx_cY$ denotes computational indistinguishability. 

\subsection{Freivalds' Algorithm}\label{subsec:prelim:freivalds}

Given matrices $\bA,\bB,\bC$, Freivalds' algorithm~\cite{Freivalds1977ProbabilisticMC} allows a verifier to efficiently check whether $\bA\bB\stackrel{?}{=}\bC$ without explicitly computing the matrix product $\bA\bB$. The verifier samples a uniformly random vector $\br$ of the appropriate dimension and checks whether
$(\br^\top \bA)\bB \stackrel{?}{=} \br^\top \bC$.
For square $n\times n$ matrices, this requires only matrix--vector multiplications, costing $O(n^2)$ field operations, whereas explicitly computing $\bA\bB$ costs $O(n^3)$ operations using the standard algorithm. If $\bA\bB=\bC$, the check always accepts; otherwise, it accepts with probability at most $|\F|^{-1}$.

For verifying a matrix--vector multiplication $\bA\mathbf{b}\stackrel{?}{=}\mathbf{c}$, Freivalds' algorithm does not asymptotically improve the verifier's work. Computing $(\br^\top \bA)\bb\stackrel{?}{=}\br^\top \bc$ requires $O(n^2)$ field operations, equivalent to directly computing $\bA\mathbf{b}$.

\subsection{Linear Error-Correcting Codes}\label{subsec:prelim:codes}
An error-correcting code encodes messages as fixed-length vectors, called codewords, by adding redundancy that enables error detection or correction; it is \textit{linear} when its codewords form a linear subspace.
Mathematically, let $n<N$ be integers and let $\bG\in\F^{n\times N}$ have linearly independent rows.
The corresponding linear code is $\calC=\calC(\bG)\coloneqq\{\bG^\top\bv:\bv\in\F^n\}\subseteq\F^N$, and $\bG$ is its generator matrix.
The parameters $n$ and $N$ are the dimension and block length, respectively, and $R(\calC)=n/N$ is the rate of $\calC$.
The minimum distance $d$ of a linear code $\calC$ is the minimum Hamming weight of a nonzero codeword.
The \emph{relative distance} of $\calC$ is $\delta=d/N$.

We say that a family of codes $\{\calC_n\}_{n=1}^\infty$ over a field $\F$ is \emph{asymptotically good} if both the code rate $R$ and the relative distance $\delta$ are constant for large enough $n$ (see, e.g. ~\cite{essentialcodingtheory}). Such codes provide efficiency in our constructions as they enable maximal error detection under minimal redundancy.

\subsection{Learning Parity with Noise}
\label{subsec:prelim:lpn}

Let $\calG$ be a distribution over code generators $\bG\in\F^{n\times N}$ and let $\calE$ be a noise distribution over sparse vectors $\be\in\F^N$. Informally, we say the $(\calG,\calE,\F)$-dual-LPN assumption holds for dimensions $n$ and $N$ if
  $  (\bG,\bG\be)\approx_c (\bG,\bu)$
where $\bG\leftarrow \calG,\be\leftarrow \calE$, and $\bu\leftarrow \F^n$. We say $\calG$ is \textit{dual-LPN-friendly} if this assumption is believed to hold for the exact noise distributions $\calE=\calE_t^N$ for a suitable sparsity parameter $t$. A formal definition of the LPN assumption and a discussion of plausibly dual-LPN-friendly code families is provided in~\cref{app:lpn}.

\section{Model and Security Definitions}\label{sec:model}

We study cryptographic primitives for private and verifiable delegation of matrix--vector products, and their applications for LLM inference.
The core primitives, \ac{vMVMD} and \ac{pvMVMD}, involve a \textit{client} holding private inputs $\bx$, and an untrusted \textit{server}.
In both primitives, the client wishes to obtain $\by=\bM \cdot \bx$ from the server for a given public matrix $\bM$.
\ac{vMVMD} allows the client to verify the correctness of $\by$, and \ac{pvMVMD} additionally hides $\bx$ from the server.

Before defining their security guarantees, we explain what it means for these primitives to be efficient and field-agnostic.
For a matrix $\bM\in\F^{m\times n}$ and vectors $\bx\in\F^n$ and $\by\in\F^m$, computing $\by=\bM\bx$ locally takes $O(mn)$ time.
We call a protocol \textbf{\textit{efficient}} if it runs in time $o_{\lambda,\kappa}(mn)$ and \textbf{\textit{asymptotically time-optimal}} if it runs in time $O_{\lambda,\kappa}(m+n)$, matching the time required to process the input and output vectors $\bx$ and $\by$.

We call a protocol \textit{field-agnostic} if it treats $\F$ as a black box, allowing field elements to have arbitrary labels and accessing them only through an oracle for field operations.
As in~\cite{EMVP25}, we assume that the oracle provides access to addition, subtraction, multiplication, inversion, unit, zero testing, and uniform sampling of field elements.

\subsection{Verifiable Matrix--Vector Multiplication Delegation (vMVMD)}
\label{subsec:vMVMD}

In \acf{vMVMD}, a client would like to efficiently verify if a server's claim $\by\stackrel{?}{=}\bM\bx$ is correct. We formalize this primitive in \cref{def:vMVMD}.

A \ac{vMVMD} protocol consists of two algorithms: $\textsc{Preprocess}$, which takes the matrix $\bM$ and outputs public parameters $\pp$, and $\textsc{Verify}$, which uses $\pp$ to verify input-output pairs $(\bx,\by)$.

Our soundness definition guarantees information-theoretic (statistical) security. In particular, it is not only computationally hard for an adversary to find a pair $(\bx,\by)$ that fools $\textsc{Verify}$, but furthermore, there does not exist any such pair $(\bx,\by)$ that fools $\textsc{Verify}$ with non-negligible probability taken over the random coins within the $\textsc{Verify}$ algorithm.

\begin{definition}[Verifiable Matrix--Vector Multiplication Delegation]\label{def:vMVMD} A \ac{vMVMD} protocol $\Pi$ is a pair of PPT algorithms $(\textsc{Preprocess}, \textsc{Verify})$ defined by:
\begin{itemize}[leftmargin=*]
    \item $\textsc{Preprocess}(1^\kappa,\bM)\rightarrow \pp$. On input statistical parameter $\kappa$ and matrix $\bM\in\F^{m\times n}$, outputs public parameters $\pp.$
    \item $\textsc{Verify}(\pp,\bx,\by)\rightarrow \{0,1\}.$ On input public parameters $\pp$, input vector $\bx\in\F^n$, and output vector $\by\in\F^m$, either accepts (outputs 1) or rejects (outputs 0).
\end{itemize}
The protocol should satisfy the following properties:
\begin{itemize}[leftmargin=*]
    \item \textbf{Completeness}: $\Pi$ is {complete} if for all $\bM \in \F^{m \times n}$ and $\bx \in \F^n$,
    \[
    \Pr \left[
        b = 1 ~\middle|~ \begin{aligned}
            &\pp\leftarrow \textsc{Preprocess}(1^\kappa,\bM) \\
            &\by=\bM\bx\\
            & b \leftarrow \textsc{Verify}(\pp,\bx,\by)
        \end{aligned}
        \right] =1.
    \]
    \item \textbf{Soundness}: $\Pi$ is {sound} if for all unbounded adversaries $\adv$ there exists a negligible function $\negl(\cdot)$ such that for all $\bM\in\F^{m\times n}$ and $\bx\in\F^n$\footnote{The adversary is unbounded, providing information-theoretic security.},

    \[
    \resizebox{0.93\columnwidth}{!}{$
    \Pr \left[
        \by \ne \bM \bx \;\land\; b = 1
        ~\middle|~ \begin{aligned}
            &\pp\leftarrow\textsc{Preprocess}(1^\kappa,\bM) \\ &\by \leftarrow \adv(\pp, \bM,\bx) \\
            & b \leftarrow \textsc{Verify}(\pp,\bx,\by)
        \end{aligned}\right] < \negl(\kappa).
        $}
    \]
\end{itemize}
\end{definition}

\subsection{Private and Verifiable Matrix--Vector Multiplication Delegation (pvMVMD)}
In private and verifiable matrix--vector multiplication delegation (\ac{pvMVMD}), a client would like to outsource the computation of $\bM \bx$ to the server without revealing $\bx$. The client encrypts $\bx$ to obtain $\widehat{\bx}$, the server evaluates $\widehat{\by}=\bM\widehat{\bx}$, and the client verifies and decrypts the result. The server may deviate from the prescribed protocol, and a public verification algorithm provides soundness for the client against such deviations.
We formalize this primitive in \cref{def:pvMVMD}.

\begin{definition}[Private and Verifiable Matrix--Vector Multiplication Delegation]\label{def:pvMVMD}
A \ac{pvMVMD} protocol $\Pi$ is a tuple of PPT algorithms $\textsc{(Preprocess,Encrypt,Verify,Decrypt)}$ defined by:
\begin{itemize}[leftmargin=*]
    \item $\textsc{Preprocess}(1^\lambda,1^\kappa,\bM)\rightarrow \pp$. On input security parameter $\lambda$, statistical parameter $\kappa$, and matrix $\bM\in\F^{m\times n}$, outputs public parameters $\pp$.
    \item $\textsc{Encrypt}(\pp,\bx)\rightarrow (\widehat{\bx},\st)$. On input public parameters $\pp$ and plaintext vector $\bx\in\F^n$, outputs the ciphertext $\widehat{\bx}$ and a private state $\st$.
    \item $\textsc{Verify}(\pp, \widehat{\bx},\widehat{\by})\rightarrow \{0,1\}.$ On input public parameters $\pp$, ciphertext $\widehat{\bx}\in\F^n$, and output vector $\widehat{\by}\in\F^m$, either accepts (outputs 1) or rejects (outputs 0).
    \item $\textsc{Decrypt}(\pp,\st,\widehat{\by})\rightarrow \by.$ On input public parameters $\pp$, private state $\st$, and ciphertext $\widehat{\by}\in\F^m$, outputs the plaintext vector $\by$ which is claimed to equal $\bM\bx.$
\end{itemize}

The protocol should satisfy the following properties:
\begin{itemize}[leftmargin=*]
    \item \textbf{Completeness}: $\Pi$ is complete if, whenever $\Pi$ is executed honestly by both the client and the server, then the client outputs $\by=\bM\bx$ on input $\bx$ and public matrix $\bM$.
    Formally, for any $\bM\in\F^{m\times n}$ and $\bx \in \F^{n}$,
\[
\resizebox{0.8\columnwidth}{!}{$
\Pr\left[\by = \bM \bx \;\land\;  b=1
\;\middle|\;\begin{aligned}
    &\pp\leftarrow\textsc{Preprocess}(1^\lambda,1^\kappa,\bM)\\&(\widehat{\bx},\st)\leftarrow\textsc{Encrypt}(\pp,\bx)\\
    &\widehat{\by}=\bM\widehat{\bx}
    \\
    & b \leftarrow  \textsc{Verify}(\pp,\widehat{\bx},\widehat{\by})\\
    &\by\leftarrow \textsc{Decrypt}(\pp,\st,\widehat{\by}) 
\end{aligned}\right]
=1.
$}
\]

\item \textbf{Soundness}: $\Pi$ is sound if a computationally unbounded server $\adv$ that returns a vector $\widehat{\by}\neq \bM\widehat{\bx}$ will be rejected by $\textsc{Verify}$ with high probability.
Formally, there exists some negligible function $\negl(\cdot)$ such that for any $\bM\in\F^{m\times n}$ and $\bx \in \F^n$,
\[
\resizebox{0.9\columnwidth}{!}{$
\Pr\left[
\by\ne \bM\bx\;\land\; b=1
\;\middle|\;\begin{aligned}
    &\pp\leftarrow\textsc{Preprocess}(1^\lambda,1^\kappa,\bM)\\&(\widehat{\bx},\st)\leftarrow\textsc{Encrypt}(\pp,\bx)\\& \widehat{\by}\leftarrow\adv(\pp,\bM,\widehat{\bx}) \\& b \leftarrow \textsc{Verify}(\pp,\widehat{\bx},\widehat{\by})
    \\&\by\leftarrow \textsc{Decrypt}(\pp,\st,\widehat{\by})  
\end{aligned}\right]
<\negl(\kappa).
$}
\]

\item \textbf{Client privacy}: $\Pi$ is private for the client if executing $\Pi$ reveals no information about $\bx$ to the server.
Formally, there exists a PPT simulator $\mathsf{Sim}$ such that for any matrix $\bM\in\F^{m\times n}$ and $\bx\in\F^n$, the following distributions are computationally indistinguishable with respect to $\kappa$ and $\lambda$:
\[
\mathsf{Sim}(1^\lambda,1^\kappa, \bM) \approx_c \mathsf{View}_\textnormal{Server}^{\Pi}(\bM,\bx),
\]
where $\mathsf{View}_\textnormal{Server}^{\Pi}(\bM,\bx)\coloneqq(\pp,\bM,\widehat{\bx})$ denotes the server's view, given by $\pp\leftarrow \textsc{Preprocess}(1^\lambda,1^\kappa,\bM)$ and $(\widehat{\bx},\st)\leftarrow \textsc{Encrypt}(\pp,\bx).$

\end{itemize}
\end{definition}

\begin{remark}[Multi-query privacy]
Client privacy extends to multiple queries via a standard hybrid argument, and the multi-query distinguishing advantage for $q$ queries is at most $q$ times the single-query advantage.
\end{remark}

\begin{remark}[Composition of vMVMD with pvMVMD] We remark that any possibly unsound pvMVMD protocol can be made sound by using the $\textsc{Verify}$ algorithm of a sound vMVMD protocol. A formal proof appears in \cref{app:soundness-composition}.
\end{remark}

\subsection{Transparent preprocessing}
We call the preprocessing \emph{transparent} if its output contains no secret trapdoor or client-specific verification key and can be validated against the public matrix $\bM$ by any party.
The server or a third party may therefore generate and publish the preprocessing state once for all clients delegating the same matrix.
A client that does not trust the publisher can independently validate the state at a cost comparable to preprocessing and then reuse it for all future queries involving $\bM$.
Alternatively, the publisher may attach a proof of correct preprocessing to reduce this per-client validation cost.

\section{Verifiable Matrix--Vector Multiplication Delegation (vMVMD)}
\label{sec:verification}

In this section, we present a novel asymptotically time-optimal \ac{vMVMD} protocol using tools from coding theory.
We present our construction in \cref{subsec:verification:construction}, demonstrate its security in \cref{subsec:verification:security}, analyze its costs in \cref{subsec:verification:efficiency}, and extend it with a batch verification protocol in \cref{sec:verification:batch-verification}, tailored to the batched \ac{mvm} claims arising in our LLM application.

\subsection{Our Construction}\label{subsec:verification:construction}

\begin{protocol}[!t]
\begin{mdframed}
\small
\setlist[itemize]{nosep}
\begin{center}
\colorbox{gray!20}{$\PivMVMD$}
\end{center}
\underline{\textbf{Input parameters}}:
\vspace{0.1em}
\begin{itemize}[leftmargin=*]
    \item Distribution $\calG$ over generator matrices $\bG \in \F^{m \times N}$ of codes with rate $R=\Theta(1)$, relative distance $\delta=\Theta(1)$, and linear encoding time.
    \item Sparsity parameter $t=t(\kappa)$ and repetition parameter $\ell=\ell(\kappa)$.
    \item Block length $N=O(m).$
    \item Field oracle $\F$.
\end{itemize}

\vspace{0.5em}
\underline{\textbf{\textsc{Preprocess}}$(1^\kappa,\bM)\rightarrow \pp$}:
\vspace{0.1em}
\begin{itemize}[leftmargin=*]
    \item Sample $\bG\leftarrow \calG$.
    \item Compute $\bQ\coloneqq\bG^\top \bM\in\F^{N\times n}$ and output $\pp = (\bG, \bQ,  t,\ell)$.
\end{itemize}

\vspace{0.5em}
\underline{\textbf{\textsc{Verify}}$(\pp,\bx,\by)\rightarrow \{0,1\}$}:
\vspace{0.1em}
\begin{itemize}[leftmargin=*]
    \item Compute $\bG^\top \by\in\F^N$
    \item Repeat $\ell$ times:
    \begin{itemize}[leftmargin=*]
        \item Sample $\be\leftarrow \calE_{t}^{N}$.
        \item Compute $v_L\coloneqq(\be^\top \bQ)\bx$ and $v_R\coloneqq\be^\top (\bG^\top \by)$
        \item If $v_L\neq v_R$, output 0 (reject) and halt.
    \end{itemize}
    \item If all $\ell$ checks pass, output 1 (accept).
\end{itemize}
\end{mdframed}
\caption{Verifiable matrix--vector multiplication delegation protocol}
\label{prot:vMVMD}
\end{protocol}

We begin by modifying Freivalds' algorithm to enable faster client verification of \acp{mvm}. Recall that $\calE_{t}^{m}$ denotes the set of $t$-sparse vectors of dimension $m$.
Rather than sampling a Freivalds' challenge uniformly at random, we sample a $t$-sparse vector $\be\in \calE_{t}^{m}$.
Given a verification claim
$\bM\bx\stackrel{?}{=}\by,$
the client can compute $(\be^\top\bM)\bx$ and $\be^\top\by$ in time $O(tn)$ and accept if and only if they match.
When $t = O(\kappa)$, we have a protocol that runs in $O_\kappa(m + n)$ to verify instead of $O(mn)$.

However, the soundness of this protocol is not ideal.
Consider an invalid expression $\bM\bx\neq \by$, and define the \emph{error vector} $\bDelta\coloneqq\bM\bx-\by\in\F^m$.
The previous protocol accepts whenever $\be^\top \bDelta=0$. For instance, if there is exactly one nonzero entry in $\bDelta$ (i.e., $\wt(\bDelta) = 1$), then the protocol rejects only if $\be$ is nonzero in the same coordinate as $\bDelta$, which occurs with probability $t/m$. This yields a soundness error of at least $1-t/m$.

To reduce the soundness error, we use \emph{linear} codes (see \cref{subsec:prelim:codes}) for dimension $m$ and a block length $N$ with the following two properties: (1) the code has \emph{high relative distance} $\delta$ when the rate $R=m/N$ is kept constant for security, and (2) the code has \emph{fast encoding time} for efficiency.
Intuitively, the encoding spreads every nonzero error vector $\bDelta$ into a codeword $\bG^\top\bDelta\in\F^N$ with at least $\delta N$ nonzero entries.

Let $\bG\in\F^{m\times N}$ be the generator matrix of such a linear code.
Rather than verifying the original claim directly, we verify its encoding under $\bG^\top$.
That is, the client samples a sparse challenge $\be\in\calE_t^N$ and accepts if $(\be^\top\bG^\top\bM)\bx=\be^\top\bG^\top\by$.

If $\bDelta=\bM\bx-\by\neq 0$, the minimum distance of the code guarantees that $\bG^\top\bDelta\in\F^N$ has at least $\delta N$ nonzero entries.
Thus, the support of a uniformly random $t$-sparse vector $\be$ is likely to intersect the support of $\bG^\top\bDelta$.

Computing $\bG^\top\bM$ for every query incurs heavy work, but we recognize that the term $\bQ=\bG^\top \bM$ is independent of client input, and thus can be prepared in a preprocessing phase.
Anyone can verify the correctness of $\bQ$ since both $\bG$ and $\bM$ are public; our preprocessing phase adds no extra trust assumption.
The full protocol is formally presented in \cref{prot:vMVMD}.

Meanwhile, $(\be^\top \bQ)\bx$ and $\be^\top (\bG^\top \by)$ can be computed efficiently, since $\be^\top \bQ$ can be computed in time proportional to the sparsity of $\be$ and the fast encoding time of $\bG$ makes computing $\bG^\top \by$ efficient.

\subsection{Security}
\label{subsec:verification:security}
We show that $\PivMVMD$ satisfies soundness for suitable parameters $t$, $\ell$, and $\delta$.
We first prove the following lemma.

\begin{lemma}\label{lemma:vMVMD-soundness}
    Let $\bQ\coloneqq\bG^\top\bM\in\F^{N\times n}$, and let $\be$ be sampled uniformly from $\calE_t^N$.
    If the code generated by $\bG$ has relative distance $\delta$, then for any $\bx\in\F^n$ and $\by\in\F^m$ such that $\bM\bx\neq \by$, we have 
    $\Pr\left[\be^\top (\bG^\top \by) = (\be^\top \bQ)\bx \right] \leq \left(1-\delta\right)^t
    +
    \frac{1}{|\F|-1}$.
\end{lemma}
\begin{proof}
    Let \(\bu\coloneqq\by-\bM\bx\in\F^m\) be the nonzero error vector and \(\bz\coloneqq\bG^\top\bu\in\F^N\) its encoding.
    The event $\be^\top(\bG^\top\by)=(\be^\top\bQ)\bx$ is equivalent to $\be^\top\bz=0$.
    Since $\bG$ generates a code of relative distance $\delta$, we have $\wt(\bz)\ge N\delta$. 

    Let $S_{\be}\coloneqq\supp(\be)$ and $S_{\bz}\coloneqq\supp(\bz)$.
    We bound $\Pr[\be^\top\bz=0]$ by splitting into two cases based on whether $S_{\be}$ and $S_{\bz}$ intersect and applying the law of total probability.
    \begin{align*}
        &\Pr[\be^\top \bz=0]\\
        =& \Pr[\be^\top \bz=0 \mid S_{\be}\cap S_{\bz}=\emptyset]\Pr[S_{\be}\cap S_{\bz}=\emptyset]\\
        +& \Pr[\be^\top \bz=0 \mid S_{\be}\cap S_{\bz}\neq\emptyset]\Pr[S_{\be}\cap S_{\bz}\neq\emptyset]\\
        \le& \Pr[S_{\be}\cap S_{\bz}=\emptyset] + \Pr[\be^\top \bz=0 \mid S_{\be}\cap S_{\bz}\neq\emptyset]
    \end{align*}

    \underline{Term 1.} By assumption, $|S_{\bz}|\geq N\delta$ and $|S_{\be}|=t.$ Since $S_{\be}$ is a uniformly random size-$t$ subset of $[N]$ and $S_{\bz}$ is fixed, then,
    \[
    \resizebox{0.97\linewidth}{!}{$
    \Pr[S_{\be}\cap S_{\bz}=\emptyset] 
    \le
    \frac{\binom{N-N\delta}{t}}{\binom{N}{t}}
    =
    \prod_{i=0}^{t-1} \frac{N-N\delta-i}{N-i}
    \le
    \left(\frac{N-N\delta}{N}\right)^t
    =
    (1-\delta)^t.$}
    \]

    \underline{Term 2.} Define the multivariate polynomial $f_{\bz}(\be)=\be^\top \bz$.
    Conditioned on $S_{\be}\cap S_{\bz}\neq\emptyset$, $f_{\bz}$ is nonzero and has total degree 1 in the variables $\{e_j\}_{j\in S_{\be}}$. These variables are sampled uniformly at random from the set $\F\setminus\{0\}$. By the Schwartz-Zippel lemma, the probability that this polynomial evaluates to $0$ is at most $\Pr[f_{\bz}(\be)=0 \mid S_{\be}\cap S_{\bz}\neq\emptyset] \leq \frac{1}{|\F|-1}.$

    Combining both cases yields the final probability.
\end{proof}

The soundness bound of \cref{lemma:vMVMD-soundness} can be amplified using independent repetitions.
We choose $t$ so that the two terms in the single-iteration bound are comparable, and repeat $\ell$ times to obtain soundness error at most $2^{-\kappa}$.
We show that $\PivMVMD$ satisfies \cref{def:vMVMD} with the following theorem.

\begin{theorem}
\label{thm:vMVMD-secure}
If the code generated by $\bG$ has relative distance $\delta$, then the protocol $\PivMVMD$ satisfies completeness and soundness as defined in \cref{def:vMVMD} for $|\F|\geq 4$, $\ell =   \left\lceil     \frac{\kappa}{\log_2(|\F|-1)-1} \right\rceil$, and $t = \left\lceil\frac{\max(\log_2(|\F|-1), \kappa + 1)}
       {-\log_2(1-\delta)}
        \right\rceil$.
\end{theorem}
\begin{proof}
Completeness is immediate.
For any $\bM\in\F^{m\times n}$ and $\bx\in\F^n$, let $\by=\bM\bx$.
Then, for every $t$-sparse vector $\be\in\F^N$,
\[
    (\be^\top\bQ)\bx
    =
    \be^\top\bG^\top\bM\bx
    =
    \be^\top\bG^\top\by.
\]

For soundness, fix any unbounded adversary $\adv$, matrix $\bM\in\F^{m\times n}$, and vector $\bx\in\F^n$.
Condition on any public parameters $\pp$ output by \textsc{Preprocess} for which the code generated by $\bG$ has relative distance $\delta$, and on any output $\by \leftarrow \adv(\pp,\bM,\bx)$ such that $\bM\bx\neq\by$.
By \cref{lemma:vMVMD-soundness}, each iteration of \textsc{Verify} accepts with probability at most $(1-\delta)^t+\frac{1}{|\F|-1}$.
The protocol accepts only if all $\ell$ iterations pass.
Since the iterations sample their challenges $\be$ independently, the total acceptance probability is at most $\left((1-\delta)^t+\frac{1}{|\F|-1}\right)^\ell$.

We now consider the two parameter spaces.

\textbf{Case 1:}
Suppose $\frac{1}{|\F|-1}\leq 2^{-(\kappa+1)}$.
The choice of $t$ gives $(1-\delta)^t \leq 2^{-(\kappa+1)}$.
Since $\ell=1$, the soundness error is at most
\[
    (1-\delta)^t+\frac{1}{|\F|-1}
    \leq
    2^{-(\kappa+1)}+2^{-(\kappa+1)}
    =
    2^{-\kappa}.
\]

\textbf{Case 2:}
Suppose $\frac{1}{|\F|-1}>2^{-(\kappa+1)}$.
The choice of $t$ gives $(1-\delta)^t \leq \frac{1}{|\F|-1}$.
Therefore, the soundness error is at most
    $\left(
        \frac{2}{|\F|-1}
    \right)^\ell
    =
    2^{-\ell(\log_2(|\F|-1)-1)}$.
By the choice of $\ell$, then $\ell\bigl(\log_2(|\F|-1)-1\bigr) \geq \kappa$, and hence
    $\left(
        \frac{2}{|\F|-1}
    \right)^\ell
    \leq
    2^{-\kappa}$.

The bound holds for every public parameter $\pp$ satisfying the relative distance condition and every incorrect output chosen by $\adv$. We therefore obtain
\[
    \Pr \left[
        \by \ne \bM \bx \;\land\; b = 1
        ~\middle|~ \begin{aligned}
            &\pp\leftarrow\textsc{Preprocess}(1^\kappa,\bM) \\
            &\by \leftarrow \adv(\pp, \bM,\bx) \\
            & b \leftarrow \textsc{Verify}(\pp,\bx,\by)
        \end{aligned}\right] \leq 2^{-\kappa}.
\]
Thus, $\PivMVMD$ satisfies \cref{def:vMVMD}.
\end{proof}

\begin{remark}[Composing with code failure probability]
\label{remark:pfail-composition}
Suppose $\bG\leftarrow\calG$ is sampled from a code distribution and let $\pfail=\pfail(\calG,N,\delta)$ be the probability that the sampled code has relative distance smaller than $\delta$, then the overall soundness error of $\PivMVMD$ is at most $\pfail+2^{-\kappa}$, which remains negligible when $\pfail$ is negligible. Thus, for code families where $\pfail$ can be made arbitrarily small under practical parameters, the protocol satisfies \cref{def:vMVMD}.
\end{remark}

\subsection{Efficiency}
\label{subsec:verification:efficiency}

\begin{table}[htbp]
\centering
\small
\begin{tabular}{lc}
\toprule
\textbf{Operation} & \textbf{Cost} \\
\midrule
$\bG^\top \by$ & $O(N)$ \\
$\be^\top (\bG^\top \by)=\sum_{j\in \supp(\be)}e_j(\bG^\top \by)_j$ & $2t-1$ \\
$\be^\top \bQ=\sum_{j\in \supp(\be)}e_j \bQ_{j,\ast}$ & $(2t-1)n$  \\
$(\be^\top \bQ)\bx$ & $2n-1$ \\
\midrule
\textbf{$\ell$ repetitions total} & $O\left(N+\ell tn\right)$ \\
\bottomrule
\end{tabular}
\caption{Cost tally for the client online phase of $\PivMVMD$.}
\label{tab:vMVMD-operations}
\end{table}

For optimal efficiency, we instantiate $\PivMVMD$ with a code with linear encoding time, constant rate, and constant distance.
Under the parameters for $\ell$ and $t$ specified in \cref{thm:vMVMD-secure}, $\ell\cdot t=O(\kappa)$, so $\PivMVMD$ has asymptotic time $O_{\kappa}(m + n)$.
The verification cost is summarized in \cref{tab:vMVMD-operations}.
The cost is asymptotically time-optimal, since reading the input $\bx$ and output $\by$ takes $O(m+n)$ time. 

\parhead{Preprocessing time.} The preprocessing phase computes $\bQ=\bG^\top \bM\in\F^{N\times n}$, i.e., $n$ encodings of the columns of $\bM$.
Since the encoding time is linear with linear-encodable codes, $\textsc{Preprocess}$ has running time $O(mn)$.

\subsection{Batch Verification}
\label{sec:verification:batch-verification}
In our LLM application, a single inference produces many \ac{mvm} claims that must eventually be verified. The same setting arises more generally whenever a client verifies computations under a fixed collection of matrices. Suppose the client wishes to verify $\bM_b\bx_b\stackrel{?}{=}\by_b$ for $b\in[B]$, where $\bM_b\in\F^{m\times n}$; claims with different dimensions can be grouped by shape and verified in separate batches.

Recall that preprocessing computes $\bQ_b\coloneqq\bG^\top\bM_b$. Within each verification repetition, the same sparse challenge $\be$ can be shared across all $B$ claims without increasing the soundness error. However, the client must still compute $\bQ_b^\top\be$ for every claim and every repetition, costing $O(B\ell tn)$ field operations, which is usually the bottleneck. Our first observation is that these challenge-dependent projections can themselves be delegated to the server.

For one challenge $\be$, define $\bu_b\coloneqq\bQ_b^\top\be$ and stack all the $B$ computations into the single relation $\widetilde{\bu}=\widetilde{\bQ}\be$ where $\widetilde{\bQ}\coloneqq \left[\bQ_1\;|\;\bQ_2\;|\ldots\;|\bQ_B\right]^\top$.
We ask the server to compute $\widetilde{\bu}$ and verify this auxiliary relation using a second code $\bG'\in\F^{Bn\times N'}$. Preprocessing computes
$\widetilde{\bQ}^{\mathsf{aux}}\coloneqq\widetilde{\bQ}^{\top}\bG'.$
After the server fixes its auxiliary response, the client samples a fresh $t'$-sparse challenge $\be'$, computes $\br\coloneqq\bG'\be'$, and checks $\be^\top\widetilde{\bQ}^{\mathsf{aux}}\be'\stackrel{?}{=}\br^\top\widetilde{\bu}.$ The full protocol repeats this primary check $\ell$ times and uses $\ell'$ independent secondary repetitions.

Several observations make this second verification much more efficient. Unlike the original challenge projection, this auxiliary relation can be checked directly by the client: because both $\be$ and $\be'$ are sparse, the left-hand side accesses only a $t\times t'$ submatrix of $\widetilde{\bQ}^{\mathsf{aux}}$, while the right-hand side is a single inner product of length $Bn$. Thus, no further delegation is needed. Moreover, the full protocol runs $\ell$ independent primary repetitions and $\ell'$ independent secondary repetitions. Within each secondary repetition, the same challenge $\be'$ can be used to check all $\ell$ auxiliary responses, since all of them are fixed before $\be'$ is sampled. Consequently, $\bG'\be'$ is computed only once per secondary repetition, although the resulting scalar check is performed against each primary response. 

The resulting protocol adds one challenge--response round and moves the expensive projection work from the client to the server, reducing the client's online time from $O(BN+B\ell t n)$ naively to
$O\!\left(BN+B\ell(n+t)+\ell't'Bn+\ell\ell'(Bn+tt')\right)$. We give the complete protocol in \cref{prot:batch-vMVMD}. In the following, we analyze this protocol.

We first prove the security of $\Pi_{\textnormal{Batch-vMVMD}}$. The proof follows the ideas in the primary $\PivMVMD$ analysis.

\begin{theorem}[Security of batch verification]
\label{thm:batch-vMVMD-secure}
Suppose $\bG$ and $\bG'$ generate codes of relative distances at
least $\delta$ and $\delta'$, respectively. Let
$\epsilon\coloneqq(1-\delta)^t+\frac{1}{|\F|-1}$ and
$\epsilon'\coloneqq(1-\delta')^{t'}+\frac{1}{|\F|-1}$.
Then $\Pi_{\mathsf{Batch\text{-}vMVMD}}$ has perfect completeness
and soundness error at most
\[
    \epsilon^\ell+(\epsilon')^{\ell'}.
\]
In particular, the parameters can be chosen so that the soundness
error is at most $2^{-\kappa}$.
\end{theorem}

\begin{proof}
Completeness follows directly from
$\bu_b^{(q)}=\bQ_b^\top\be^{(q)}$ and
$\widetilde{\bu}^{(q)}=\widetilde{\bQ}\be^{(q)}$.

For soundness, first suppose every auxiliary response is correct.
Fix any incorrect original claim $b^\star$ and let
$\boldsymbol{\Delta}_{b^\star}
\coloneqq\by_{b^\star}-\bM_{b^\star}\bx_{b^\star}\neq\mathbf 0$.
Acceptance in primary repetition $q$ requires
$(\be^{(q)})^\top\bG^\top\boldsymbol{\Delta}_{b^\star}=0$.
By \cref{lemma:vMVMD-soundness}, this occurs with probability at
most $\epsilon$ per repetition, and hence with probability at most
$\epsilon^\ell$ over all $\ell$ independent repetitions. Sharing
$\be^{(q)}$ across the $B$ claims introduces no factor $B$, since
batch acceptance requires this fixed incorrect claim to accept.

Otherwise, fix any incorrect auxiliary response
$q^\star\in[\ell]$ and define
$\boldsymbol{\Gamma}^{(q^\star)}
\coloneqq\widetilde{\bu}^{(q^\star)}
-\widetilde{\bQ}\be^{(q^\star)}\neq\mathbf 0$.
This vector is fixed before the secondary challenges are sampled.
Acceptance in secondary repetition $j$ requires
$({\be'}^{(j)})^\top\bG'^\top
\boldsymbol{\Gamma}^{(q^\star)}=0$, which again occurs with
probability at most $\epsilon'$ by
\cref{lemma:vMVMD-soundness}. Thus all $\ell'$ secondary
repetitions accept with probability at most
$(\epsilon')^{\ell'}$. Sharing each secondary challenge across the
$\ell$ auxiliary responses introduces no factor $\ell$, since batch
acceptance requires the fixed incorrect response $q^\star$ to
accept.

Combining the two cases gives the claimed bound.
\end{proof}

If either code is sampled from a distribution that may have relative
distance below its target value, the corresponding code-distance
failure probabilities are added to the above bound, as in
\cref{remark:pfail-composition}.

\parhead{Client online time.}
The client first computes $\bz_b\coloneqq\bG^\top\by_b$ for all $b\in[B]$,
costing $O(BN)$ using the primary linear-time encoder. The original
checks then cost $O(B\ell(n+t))$.

For each secondary repetition $j$, since ${\be'}^{(j)}$ is
$t'$-sparse, $\bG'{\be'}^{(j)}
=\sum_{k\in\supp({\be'}^{(j)})}e'_k\bG'_{\ast,k}$ and can therefore
be computed in $O(t'Bn)$ field operations. Because the same
secondary challenge is shared across all $\ell$ primary responses,
this cost is incurred only $\ell'$ times.

For every pair $(q,j)$, computing
$(\be^{(q)})^\top\widetilde{\bQ}^{\mathsf{aux}}{\be'}^{(j)}$
costs $O(tt')$, while
$(\br^{(j)})^\top\widetilde{\bu}^{(q)}$ costs $O(Bn)$.
Hence the total online client time is
\[
O\!\left(
BN+B\ell(n+t)+\ell't'Bn+\ell\ell'(Bn+tt')
\right).
\]
Independent verification instead pays $O(B\ell tn)$ for computing
the challenge-dependent projections $\bQ_b^\top\be^{(q)}$.
Thus, batch verification moves this work to the server. For
constant-distance codes, $t,\ell,t',$ and $\ell'$ depend only on
$\kappa$, so the client time is \(O_\kappa(B(m+n))\).

The optimization adds one challenge--response round. The server
performs $O(B\ell tn)$ additional field operations and returns
$B\ell n$ field elements.

\parhead{Preprocessing and state.}
Batch verification additionally preprocesses
$\widetilde{\bQ}^{\mathsf{aux}}
\coloneqq\widetilde{\bQ}^{\top}\bG'\in\F^{N\times N'}$.
If applying the secondary encoder to a vector in $\F^{Bn}$ costs
$T_{\bG'}(Bn,N')$, this requires
$O(NT_{\bG'}(Bn,N'))$ preprocessing time and $NN'$ field elements
of additional state. For a constant-rate, linear-time encodable
secondary code, $N'=\Theta(Bn)$, so both quantities are
$O(Bmn)$.

The storage can be reduced when $\bG'$ is a systematic code. If
$\bG'=[I_{Bn}\mid\bP]$, then
$\widetilde{\bQ}^{\mathsf{aux}}
=[\widetilde{\bQ}^{\top}\mid\widetilde{\bQ}^{\top}\bP]$.
The first block is already represented by
$\bQ_1,\ldots,\bQ_B$, so only
$\widetilde{\bQ}^{\top}\bP$ needs additional storage. This
optimization reduces storage but does not inherently reduce
preprocessing time, which depends on the cost of applying the
secondary encoder.

The auxiliary preprocessing depends only on the fixed matrices in
the batch and can therefore be reused across future batches
involving the same matrices.

\begin{protocol}[!h]
\begin{mdframed}
\small
\setlist[itemize]{nosep}

\begin{center}
\colorbox{gray!20}{$\Pi_{\textnormal{Batch-vMVMD}}$}
\end{center}

\underline{\textbf{Input parameters}}:
\vspace{0.1em}
\begin{itemize}[leftmargin=*]
    \item Distribution $\calG$ over generator matrices $\bG\in\F^{m\times N}$ of codes with rate $R=\Theta(1)$, relative distance $\delta=\Theta(1)$, and efficient encoding time.
    \item Distribution $\calG'$ over generator matrices $\bG'\in\F^{Bn\times N'}$ of codes with rate $R'\coloneqq Bn/N'$, relative distance $\delta'$, and efficient encoding time.
    \item Primary sparsity and repetition parameters $t$ and $\ell$.
    \item Secondary sparsity and repetition parameters $t'$ and $\ell'$.
    \item Primary block length $N=O(m)$ and secondary block length $N'$.
    \item Field oracle $\F$.
\end{itemize}

\vspace{0.5em}
\underline{\textbf{\textsc{Preprocess}}%
$(1^\kappa,\{\bM_b\}_{b\in[B]})\rightarrow\pp$}:
\vspace{0.1em}
\begin{itemize}[leftmargin=*]
    \item Sample $\bG\leftarrow\calG$ and $\bG'\leftarrow\calG'$.
    \item For every $b\in[B]$, compute $\bQ_b\coloneqq\bG^\top\bM_b\in\F^{N\times n}$.
    \item Form $\widetilde{\bQ}\coloneqq(\bQ_1,\ldots,\bQ_B)^\top\in\F^{Bn\times N}$.
    \item Compute $\widetilde{\bQ}^{\mathsf{aux}}\coloneqq\widetilde{\bQ}^\top\bG'\in\F^{N\times N'}$.
    \item Output $\pp\coloneqq\left(\bG,\bG',\{\bQ_b\}_{b\in[B]},\widetilde{\bQ}^{\mathsf{aux}},t,\ell,t',\ell'\right)$.
    \item Make the public matrices $\{\bQ_b\}_{b\in[B]}$ available to the server.
\end{itemize}

\vspace{0.5em}
\underline{\textbf{\textsc{Verify}}%
$(\pp,\{(\bx_b,\by_b)\}_{b\in[B]})\rightarrow\{0,1\}$}:
\vspace{0.1em}
\begin{itemize}[leftmargin=*]
    \item \textbf{Client:}
    Require that all alleged outputs $\by_1,\ldots,\by_B$ are fixed before sampling any verification challenge.

    \item \textbf{Client:}
    For every $b\in[B]$, compute $\bz_b\coloneqq\bG^\top\by_b\in\F^N$.

    \item \textbf{Client:}
    Sample independent primary challenges $\be^{(1)},\ldots,\be^{(\ell)}\leftarrow\calE_t^N$ and send $\{\be^{(q)}\}_{q\in[\ell]}$ to the server.
    Within each repetition $q$, the same challenge $\be^{(q)}$ is used for all $B$ claims.

    \item \textbf{Server:}
    For every $q\in[\ell]$, compute $\widetilde{\bu}^{(q)}\coloneqq\widetilde{\bQ}\be^{(q)}\in\F^{Bn}$ and return $\{\widetilde{\bu}^{(q)}\}_{q\in[\ell]}$ to the client.

    \item \textbf{Client:}
    After all server responses are fixed, sample independent secondary challenges ${\be'}^{(1)},\ldots,{\be'}^{(\ell')}\leftarrow\calE_{t'}^{N'}$.

    \item \textbf{Client:}
    For every $j\in[\ell']$, compute $\br^{(j)}\coloneqq\bG'{\be'}^{(j)}\in\F^{Bn}$.

    \item \textbf{Client:}
    For every $q\in[\ell]$ and $j\in[\ell']$, check $(\be^{(q)})^\top\widetilde{\bQ}^{\mathsf{aux}}{\be'}^{(j)}\stackrel{?}{=}(\br^{(j)})^\top\widetilde{\bu}^{(q)}$.
    If any secondary check fails, output $0$ (reject) and halt.

    \item \textbf{Client:}
    For every $q\in[\ell]$, parse $\widetilde{\bu}^{(q)}\coloneqq((\bu_1^{(q)})^\top,\ldots,(\bu_B^{(q)})^\top)^\top$, where $\bu_b^{(q)}\in\F^n$.

    \item \textbf{Client:}
    For every $b\in[B]$ and $q\in[\ell]$, check $(\bu_b^{(q)})^\top\bx_b\stackrel{?}{=}(\be^{(q)})^\top\bz_b$.
    If any original check fails, output $0$ (reject) and halt.
    If all checks pass, output $1$ (accept).
\end{itemize}

\end{mdframed}

\caption{Batch verification for multiple matrices and vectors.
The client delegates the stacked matrix-dependent challenge projection to the server and verifies the returned vector using a second coded check.
}
\label{prot:batch-vMVMD}
\end{protocol}

\section{Private and Verifiable Matrix--Vector Multiplication Delegation (pvMVMD)}
\label{sec:delegation}

In this section, we present a construction for \ac{pvMVMD} by combining $\PivMVMD$ in \cref{sec:verification} for verification with dual-LPN-friendly linear codes for efficient privacy.
We introduce our construction in \cref{subsec:delegation:construction}, analyze its security in \cref{subsec:delegation:security}, and its costs in \cref{subsec:delegation:efficiency}.

\subsection{Our Construction}
\label{subsec:delegation:construction}
To encrypt a vector $\bx\in\F^n$, the client masks $\bx$ with a vector $\bx'$ to form the ciphertext $\widehat{\bx}=\bx+\bx'\in\F^n$. 
The client sends $\widehat{\bx}$ to the server, which returns $\widehat{\by}=\bM\widehat{\bx}$.
The client verifies the correctness of $\widehat{\by}$ by running a \ac{vMVMD} protocol over the ciphertexts.
The client decrypts $\widehat{\by}$ by computing $\by=\widehat{\by}-\bM\bx'\in\F^m$.
However, if $\bx'$ is sampled uniformly at random, the time required for the client to compute $\bM\bx'$ matches the time for the client to locally compute $\by=\bM\bx$, rendering it useless to outsource this computation.
To solve the efficiency problem, we prepare $\bx'$ pseudorandomly such that $\bM\bx'$ can be computed in time $O(m+n)$.
We make use of coding theory to construct such pseudorandom vectors.

Our construction for privacy follows the idea of prior works~\cite{abbaszadeh2025single,EMVP25,braverman2025practical}.
We use a separate distribution of linear codes $\privcalG$ for privacy, distinct from the code used by $\PivMVMD$ in \cref{sec:verification}.
Let $\privG\in\F^{n\times\privN}$ denote a sampled generator matrix, where $\privN$ is the block length of the privacy code.
Under the $(\privcalG,\calE_{\privt}^{\privN},\F)$-dual-LPN assumption, if $\prive\in\F^{\privN}$ is $\privt$-sparse then $\bx'\coloneqq\privG\prive$ is a pseudorandom vector. We mask $\bx$ by computing $\widehat{\bx}=\bx+\bx'$. By first computing $\bP\coloneqq\bM\privG\in\F^{m\times\privN}$ in a preprocessing phase, we can remove the mask by computing $\bM\bx'=\bM(\privG\prive)=(\bM\privG)\prive=\bP\prive$.
Since $\prive$ is sparse, computing $\bP\prive$ requires combining only the columns of $\bP$ corresponding to the nonzero entries of $\prive$.

The full protocol $\PipvMVMD$ is presented in \cref{prot:pvMVMD}.

The client privacy property can be strengthened to additionally hide a private matrix $\bM$ owned by the client. While such a protocol enables greater privacy, our construction requires a private client-run preprocessing, rendering this most useful when the preprocessing can be amortized across many queries.
For further details on our construction, see~\cref{app:private-matrix}.

\begin{protocol}[!t]
\begin{mdframed}
\small
\setlist[itemize]{nosep}
\begin{center}
\colorbox{gray!20}{$\PipvMVMD$}
\end{center}
\underline{\textbf{Input parameters}}:
\vspace{0.1em}
\begin{itemize}[leftmargin=*]
    \item Distribution $\privcalG$ over generator matrices $\privG \in \F^{n \times \privN}$ of codes with rate $\privR=\Theta(1)$ and relative distance $\privdelta=\Theta(1)$, linear encoding time, and satisfying the dual-LPN assumption for error distribution $\calE_{\privt}^{\privN}.$
    \item Distribution $\vercalG$ over generator matrices $\verG \in \F^{m \times \verN}$ of codes with rate $\verR=\Theta(1)$, relative distance $\verdelta=\Theta(1)$, and linear encoding time.
    \item Sparsity parameters $\privt=\privt(\lambda)$ and $\verit=\verit(\kappa)$ and repetition parameter $\verl=\verl(\kappa)$.
    \item Block lengths $\privN=O(n)$ and $\verN=O(m).$
    \item Field oracle $\F$.
\end{itemize}

\vspace{0.5em}
\underline{\textbf{\textsc{Preprocess}}$(1^\lambda, 1^\kappa, \bM) \to \pp$}:
\vspace{0.1em}
\begin{itemize}[leftmargin=*]
    \item Sample $\privG\leftarrow\privcalG$ and $\verG\leftarrow\vercalG$.
    \item Compute $\bP \coloneqq \bM\privG \in \F^{m \times \privN}$.
    \item Compute $\bQ \coloneqq \verG^\top \bM \in \F^{\verN \times n}$.
    \item Output $\pp = (\privG, \verG, \bP, \bQ, \privt, \verit,\verl)$.
\end{itemize}

\vspace{0.5em}
\underline{\textbf{\textsc{Encrypt}}$(\pp, \bx) \to (\widehat{\bx}, \st)$}:
\vspace{0.1em}
\begin{itemize}[leftmargin=*]
    \item Sample $\prive \leftarrow \calE_{\privt}^{{\privN} }$ and compute $\bx' \coloneqq \privG\prive \in \F^n$.
    \item Output $\widehat{\bx} \coloneqq \bx + \bx'$ and $\st \coloneqq \prive$.
\end{itemize}

\vspace{0.5em}
\underline{\textbf{\textsc{Verify}}$(\pp, \widehat{\bx}, \widehat{\by}) \to \{0,1\}$}:
\vspace{0.1em}

\begin{itemize}[leftmargin=*]
    \item Compute $\verG^\top \widehat{\by}\in\F^{\verN}$
    \item Repeat $\verl$ times:
    \begin{itemize}[leftmargin=*]
        \item Sample $\vere\leftarrow \calE_{\verit}^{\verN}$.
        \item Compute $v_L\coloneqq(\vere^\top \bQ)\widehat{\bx}$ and $v_R\coloneqq\vere^\top (\verG^\top \widehat{\by})$
        \item If $v_L\neq v_R$, output 0 (reject) and halt.
    \end{itemize}
    \item If all $\verl$ checks pass, output 1 (accept).
\end{itemize}

\vspace{0.5em}
\underline{\textbf{\textsc{Decrypt}}$(\pp, \st, \widehat{\by}) \to \by$}:
\vspace{0.1em}
\begin{itemize}[leftmargin=*]
    \item Output $\by \coloneqq \widehat{\by} - \bP\prive \in \F^m$.
\end{itemize}
\end{mdframed}
\caption{Private and verifiable matrix--vector multiplication
delegation protocol
}
\label{prot:pvMVMD}
\end{protocol}

\subsection{Security}
\label{subsec:delegation:security}
We now demonstrate that $\PipvMVMD$ satisfies completeness, soundness, and client privacy according to \cref{def:pvMVMD}.
For verification, we use the same parameters as in \cref{sec:verification}.
For privacy, the security analysis of prior works~\cite{expand-accumulate-codes} demonstrates that a good minimum distance suffices for the dual-LPN assumption against known attacks under certain conditions (see~\cref{app:lpn}), and they provide the following heuristic for selecting $\privt$ as a function of $\lambda, \privN$ and $\privdelta$:
\begin{equation}\label{eq:linear-tests}\privt\geq \frac{(\ln 2)(\lambda -\log_2 (\privN))}{\privdelta}.
\end{equation}
We therefore choose an asymptotically good code for $\privG$, with constant rate $\privR=\Theta(1)$, $\privN=n/\privR=O(n)$, and good relative distance $\privdelta=\Theta(1)$, guaranteeing $\privt=O_\lambda(1)$ under this heuristic. We show that $\PipvMVMD$ satisfies \cref{def:pvMVMD} in \cref{thm:pvMVMD-secure}.

\begin{theorem}
\label{thm:pvMVMD-secure}
Suppose that $\privcalG$ is a distribution over codes with relative distance at least $\privdelta$ satisfying the $(\privcalG,\calE_{\privt}^{\privN},\F)$-dual-LPN assumption for polynomially many independent samples, and that the code generated by $\verG$ has relative distance $\verdelta$.
Then $\PipvMVMD$ satisfies completeness, soundness, and client privacy as defined in \cref{def:pvMVMD} when instantiated with $|\F|\geq 4$, $\verl
    =
    \left\lceil
        \frac{\kappa}
             {\log_2(|\F|-1)-1}
    \right\rceil$, $\verit
    =
    \left\lceil
        \frac{
            \max\left(\log_2(|\F|-1),\kappa+1\right)
        }{
            -\log_2(1-\verdelta)
        }
    \right\rceil$.
\end{theorem}

\begin{proof}
Completeness is immediate.
For any $\bM\in\F^{m\times n}$ and $\bx\in\F^n$, let $\prive\leftarrow\calE_{\privt}^{\privN}$, $\bx':=\privG\prive$, and $\widehat{\bx}:=\bx+\bx'$.
If the server evaluates honestly, then $\widehat{\by}=\bM\widehat{\bx}$.
By completeness of the verification protocol, \textsc{Verify} accepts with probability 1.
Moreover, since $\bP=\bM\privG$, decryption gives
\begin{align*}
    \by
    &=
    \widehat{\by}-\bP\prive
    =
    \bM\widehat{\bx}-\bM\privG\prive\\
    &=
    \bM(\bx+\privG\prive)-\bM\privG\prive
    =
    \bM\bx.
\end{align*}

For soundness, fix any unbounded adversary $\adv$, matrix $\bM\in\F^{m\times n}$, and input $\bx\in\F^n$.
Condition on any public parameters $\pp$ output by \textsc{Preprocess} for which the code generated by $\verG$ has relative distance $\verdelta$, and on any fresh privacy randomness $\prive$ used to form $\widehat{\bx}=\bx+\privG\prive$.
Let $\widehat{\by}\leftarrow\adv(\pp,\bM,\widehat{\bx})$ be the server's response, and suppose that verification accepts.
The decrypted output satisfies $\by=\widehat{\by}-\bP\prive$.
Using $\bP=\bM\privG$, we have
\begin{align*}
    \by-\bM\bx
    &=
    \widehat{\by}
    -
    \bM\privG\prive
    -
    \bM\bx\\
    &=
    \widehat{\by}
    -
    \bM(\bx+\privG\prive)\\
    &=
    \widehat{\by}
    -
    \bM\widehat{\bx}.
\end{align*}
Therefore, if $\by\neq\bM\bx$, then $\widehat{\by}\neq\bM\widehat{\bx}$.
By \cref{thm:vMVMD-secure}, under the stated choices of $\verit$ and $\verl$, \textsc{Verify} accepts such an incorrect response with probability at most $2^{-\kappa}$, proving soundness.

For client privacy, the simulator $\mathsf{Sim}$ on input $(1^\lambda, 1^\kappa, \bM)$ runs $\pp \leftarrow \textsc{Preprocess}(1^\lambda, 1^\kappa, \bM)$, samples $\widetilde{\bx} \leftarrow \F^n$ uniformly at random, and outputs $(\pp, \bM, \widetilde{\bx})$.
By the $(\privcalG,\calE_{\privt}^{\privN},\F)$-dual-LPN assumption, $(\privG,\privG\prive)\approx_c(\privG,\bu)$, where $\privG\leftarrow\privcalG$, $\prive\leftarrow\calE_{\privt}^{\privN}$, and $\bu\leftarrow\F^n$.
Consequently, for any $\bx\in\F^n$, $(\pp,\bM,\bx+\privG\prive)\approx_c(\pp,\bM,\bx+\bu)$.
Because $\bu$ is uniformly distributed over $\F^n$, the vector $\bx+\bu$ is also identically distributed to a uniform vector $\widetilde{\bx}\leftarrow\F^n$. Hence,
\begin{align*}
    \mathsf{View}_\textnormal{Server}^{\Pi}(\bM,\bx)
    &=
    (\pp, \bM, \widehat{\bx})\\
    \approx_c
    (\pp, \bM, \widetilde{\bx})
    &=
    \mathsf{Sim}(1^\lambda, 1^\kappa, \bM).\qedhere
\end{align*}
\end{proof}

\begin{remark}[Code failure probability in $\PipvMVMD$]
As in \cref{remark:pfail-composition}, when the verification code $\verG$ is sampled from a distribution with failure probability $\pfail$, the overall soundness error is at most $\pfail+2^{-\kappa}$. Similarly, the privacy guarantee relies on the dual-LPN assumption, which requires that $\privG$ generates a code with sufficient minimum distance. Both failure probabilities can be made negligible for suitable code families that we consider for instantiation.
\end{remark}

\subsection{Efficiency}\label{subsec:delegation:efficiency}
We now analyze the online runtime of the client for $\PipvMVMD$ under the specified parameters above.

In \cref{tab:pvMVMD-operations}, we break down the runtime of each operation performed by the client in $\textsc{Encrypt}$ and $\textsc{Decrypt}$. 
\begin{table}[htbp]
\centering
\resizebox{0.97\columnwidth}{!}{\begin{tabular}{llc}
\toprule
\textbf{Phase} & \textbf{Operation} & \textbf{Cost} \\
\midrule
\textsc{Encrypt} & $\bx'\coloneqq\privG\prive$ & $O(\privN)$ \\
& $\widehat{\bx} \coloneqq \bx + \bx'$ & $n$ \\
\midrule
\textsc{Decrypt} & $\bP\prive = \sum_{j\in \supp(\prive)} e_{\bx,j} \bP_{\ast,j}$ & $(2\privt-1)m$ \\
& $\by \coloneqq \widehat{\by} - \bP\prive$ & $m$ \\
\midrule
\textsc{Verify} & See \cref{tab:vMVMD-operations} & $O\left(\verN+\verl\verit  n\right)$\\
\midrule
\textbf{Total} & -- & $O\left(\privN+\verN+n\verl\verit + m\privt\right)$ \\
\bottomrule
\end{tabular}}
\caption{Cost tally for the client online phase of $\PipvMVMD$.}
\label{tab:pvMVMD-operations}
\end{table}

Since we use the $\textsc{Verify}$ algorithm of $\PivMVMD$ (\cref{subsec:verification:construction}), $\textsc{Verify}$ runs in time $O_\kappa\left(m+n\right)$.
$\textsc{Encrypt}$ and $\textsc{Decrypt}$ have running time $O_\lambda\left(m+n\right)$, which guarantees that the total client online time of $\PipvMVMD$ is $O_{\lambda,\kappa}\left(m+n\right).$

\parhead{Preprocessing time.} The preprocessed matrix $\bP=\bM\privG$ can be computed by realizing $\bP^\top=\privG^\top \bM^\top$ as $m$ encodings of the rows of $\bM$.
Similarly, $\bQ=\verG^\top \bM$ can be realized as $n$ encodings of the columns of $\bM$.
When $\privG$ and $\verG$ have linear encoding times, $\textsc{Preprocess}$ has running time $O(mn).$

\section{Implementation and Evaluations}
\label{sec:evaluations}

We implement $\PipvMVMD$ and an end-to-end prototype of \protname.
In this section, we report on the implementation details and the evaluation results.

\subsection{Evaluation of \ac{pvMVMD}}
\label{subsec:implementation-matvec}

\parhead{Instantiation of codes.}
We implement $\PipvMVMD$ (\cref{prot:pvMVMD}) in Rust.
We use the BabyBear field $\F_q$, where $q = 15\cdot 2^{27}+1$, and set the computational and statistical security parameters to $\lambda=128$ and $\kappa=40$, respectively.

We instantiate the privacy and verification codes using the RAA code construction of Akhiani et al.~\cite{distanceRAAcodes} (also see \cref{app:RAA-codes}), with rate $\privR=\verR=\frac{1}{8}$ and target relative distance $\privdelta=\verdelta=\frac{1}{2}$. We use RAA codes because they are linear-time encodable, field-agnostic, asymptotically good, and plausibly dual-LPN-friendly. For further discussion on these properties for RAA codes, see~\cref{app:RAA-codes}.

We choose sparsity $\privt$ according to \eqref{eq:linear-tests}: for $n\in\left\{2^{10},2^{11}, 2^{12}, 2^{13}, 2^{14}\right\}$ this gives $\privt \in \left\{  160,159,157,156, 154 \right\}$,
respectively. We choose sparsity parameter $\verit = 31$ and number of repetitions $\verl = 2$ per \cref{thm:pvMVMD-secure}.

\parhead{Experimental setup}
We evaluate our \ac{pvMVMD} implementation on a research cluster with Intel Xeon Platinum 8562Y+ CPUs, using a single thread.
Unless otherwise noted, we measure computation time and exclude network communication.
For simplicity, our experiments focus on square matrices where $m=n$.

\parhead{Online Time}
The \textit{online runtime} refers to the runtime after preprocessing.
For clients, this includes \textsc{Encrypt}, \textsc{Verify}, and \textsc{Decrypt}.
\textbf{Client (total)} reports the total online time, broken down into \textbf{Client (verification)} for $\PipvMVMD.\textsc{Verify}$, and \textbf{Client (privacy)} for $\PipvMVMD.\textsc{Encrypt}$ and $\PipvMVMD.\textsc{Decrypt}$.
For servers, the online computation is the \ac{mvm} on the masked input.
The results are shown in \cref{tab:matvec-online-breakdown}.

\begin{table}[t]
\centering
\small
\setlength{\tabcolsep}{3.5pt}
\begin{tabular}{rccccc}
\toprule
\multirow{2}{*}{$\boldsymbol{n}$} & \multirow{2}{*}{\textbf{Server}}
& \multicolumn{3}{c}{\textbf{Client}} & \multirow{2}{*}{\textbf{Client speedup}} \\
\cmidrule(lr){3-5}
& & \textbf{Total} & \textbf{Verification} & \textbf{Privacy} & \\
\midrule
$2^{8}$ & $0.02$ & $0.14$ & $0.08$ & $0.06$ & $0.14\times$ \\
$2^{9}$ & $0.06$ & $0.23$ & $0.14$ & $0.09$ & $0.27\times$ \\
$2^{10}$ & $0.25$ & $0.41$ & $0.25$ & $0.15$ & $0.62\times$ \\
$2^{11}$ & $1.99$ & $0.78$ & $0.47$ & $0.31$ & $2.54\times$ \\
$2^{12}$ & $8.24$ & $1.56$ & $0.91$ & $0.65$ & $5.29\times$ \\
$2^{13}$ & $31.66$ & $3.19$ & $1.81$ & $1.38$ & $9.93\times$ \\
$2^{14}$ & $135.12$ & $7.27$ & $3.95$ & $3.31$ & $18.60\times$ \\
\bottomrule
\end{tabular}
\caption{Single-threaded online server and client runtime in $\PipvMVMD$, in milliseconds.
Values are means over three independent processes, each with ten measured online invocations after three warmups.
Client speedup is the server \ac{mvm} time divided by total client time, computed before rounding.}
\label{tab:matvec-online-breakdown}
\end{table}

For small matrices, the client online time is worse than the server time (i.e., outsourcing is not worthwhile).
However, the advantage becomes significant as matrix dimension grows. 
At $n=2^{13}$, the total client online time is \SI{3.19}{\milli\second}, consisting of \SI{1.81}{\milli\second} for verification and \SI{1.38}{\milli\second} for privacy, while the server multiplication takes \SI{31.66}{\milli\second}.
Thus, the client's online work is approximately $9.93\times$ smaller than the delegated multiplication.
At $n=2^{14}$, this gap increases to $18.60\times$, consistent with linear client and quadratic server costs.

\parhead{Preprocessing time}
\Cref{tab:matvec-preprocessing-breakdown} presents the preprocessing overhead.
Recall that preprocessing involves computing $\bP\coloneqq\bM\privG$ and $\bQ\coloneqq\verG^\top\bM$ where $\bM$ is the public matrix and $\privG,\verG$ are sampled from code generator distributions.
The computation overhead is essentially the same as encoding $\bM$.
As expected, preprocessing grows approximately quadratically with the matrix dimension because both $\bP$ and $\bQ$ require processing the complete matrix.
At $n=2^{14}$, preprocessing takes approximately \SI{272.94}{\second}, including \SI{148.64}{\second} for privacy-related preprocessing and \SI{124.30}{\second} for verification-related preprocessing.

\begin{table}[t]
\centering
\small
\begin{tabular}{rccc}
\toprule
$\boldsymbol{n}$ & \shortstack{\textbf{Preprocessing}\\\textbf{(total)}}
& \shortstack{\textbf{Preprocessing}\\\textbf{(verification)}}
& \shortstack{\textbf{Preprocessing}\\\textbf{(privacy)}} \\
\midrule
$2^{8}$ & $13.58$ & $7.42$ & $6.16$ \\
$2^{9}$ & $51.53$ & $28.65$ & $22.88$ \\
$2^{10}$ & $280.98$ & $159.34$ & $121.63$ \\
$2^{11}$ & $1{,}493.87$ & $874.19$ & $619.67$ \\
$2^{12}$ & $7{,}938.89$ & $4{,}727.21$ & $3{,}211.68$ \\
$2^{13}$ & $44{,}473.28$ & $27{,}722.33$ & $16{,}750.95$ \\
$2^{14}$ & $272{,}935.60$ & $148{,}636.17$ & $124{,}299.42$ \\
\bottomrule
\end{tabular}
\caption{One-time preprocessing runtime for square matrices in milliseconds, averaged over three independent processes.}
\label{tab:matvec-preprocessing-breakdown}
\end{table}

\parhead{Comparison with related work}
We compare our protocol with Dumas--Zucca~\cite{DumasZucca17} and Sum-Check instantiated with the BaseFold polynomial commitment scheme~\cite{thaler13,basefold}, where Dumas--Zucca is a specialized protocol for verifying matrix--vector products and Sum-Check+BaseFold represents the general proof-based approach used by recent verifiable-ML systems, including DeepProve~\cite{deepprove}.
Since both schemes provide verification only, we compare their runtime with our client verification time.
To make a fair comparison, we benchmark over the scalar field of BLS12-381 and run the experiments on an AWS \texttt{c8i.4xlarge} instance.

Across the evaluated matrix dimensions, our end-to-end time is between $2.5\times$ and $34.8\times$ faster than Dumas--Zucca and between $60.7\times$ and $194.8\times$ faster than Sum-Check+BaseFold.
At $n=4{,}096$, for example, our protocol completes in \SI{422}{\milli\second}, compared with \SI{1.056}{\second} for Dumas--Zucca and \SI{82.196}{\second} for Sum-Check+BaseFold.
Our client verification time is also comparable to theirs.
The performance numbers are shown in \cref{fig:matvec-related-work}.

\begin{figure}[t]
    \centering
    \resizebox{\linewidth}{!}{\def\mvdata{data/1.csv}%
\definecolor{mvblue}{HTML}{2A78D6}\definecolor{mvorange}{HTML}{EB6834}\definecolor{mvgreen}{HTML}{1BAF7A}%
\pgfplotsset{
  mv ours/.style={mvblue, solid, mark=*, mark options={solid, fill=mvblue}},
  mv dz/.style={mvorange, dashed, mark=square*, mark options={solid, fill=mvorange}},
  mv scbf/.style={mvgreen, densely dotted, mark=triangle*, mark options={solid, fill=mvgreen}},
}%
\begin{tikzpicture}
\begin{groupplot}[
  group style={group size=3 by 1, horizontal sep=0.3cm, yticklabels at=edge left},
  scale only axis, width=2.3cm, height=1.8cm,
  ymode=log, ymin=1e-1, ymax=4e5, ytick distance=10^2,   %
  xtick={7,...,12}, xlabel={$\log_2 n$}, xlabel style={font=\scriptsize, yshift=2pt}, xticklabel shift=-3pt,
  enlarge x limits=0.08, axis lines*=left, tick align=outside, xtick style={draw=none},
  ymajorgrids, grid style={black!12}, tick label style={font=\scriptsize},
  title style={font=\scriptsize, at={(0.5,1.0)}, anchor=south, inner sep=0pt, text height=1.5ex, text depth=0.3ex},
  every axis plot/.append style={line width=0.8pt, mark size=1.1pt},
  legend style={font=\scriptsize, draw=none, fill=none, legend columns=3, at={(-0.1,1.24)}, anchor=south west,
                /tikz/every even column/.append style={column sep=5pt}},
]
\nextgroupplot[title={Preprocessing}, ylabel={Time (ms)}, ylabel style={font=\scriptsize, yshift=-3pt}]
\addplot[mv ours] table[col sep=comma, x=logn, y=ours_preprocessing]{\mvdata};
\addplot[mv dz]   table[col sep=comma, x=logn, y=dz_preprocessing]{\mvdata};
\addplot[mv scbf] table[col sep=comma, x=logn, y=scbf_preprocessing]{\mvdata};
\legend{Ours, Dumas--Zucca, Sum-Check+BaseFold}
\nextgroupplot[title={Online, end to end}]
\addplot[mv ours] table[col sep=comma, x=logn, y expr=\thisrow{ours_server_total}+\thisrow{ours_client_online}]{\mvdata};
\addplot[mv dz]   table[col sep=comma, x=logn, y expr=\thisrow{dz_server_total}+\thisrow{dz_client_online}]{\mvdata};
\addplot[mv scbf] table[col sep=comma, x=logn, y expr=\thisrow{scbf_server_total}+\thisrow{scbf_client_online}]{\mvdata};
\nextgroupplot[title={Client verification}]
\addplot[mv ours] table[col sep=comma, x=logn, y=ours_client_verification]{\mvdata};
\addplot[mv dz]   table[col sep=comma, x=logn, y=dz_client_verification]{\mvdata};
\addplot[mv scbf] table[col sep=comma, x=logn, y=scbf_client_verification]{\mvdata};
\end{groupplot}
\end{tikzpicture}%
}%
    \caption{%
    Comparison with verifiable matrix--vector multiplication baselines for $n\times n$ matrices (log scale).}
    \label{fig:matvec-related-work}
\end{figure}

\subsection{Design and Implementation of \protname}
\label{subsec:implementation-ml}

\protname is a client-server LLM serving system where the client uses \ac{pvMVMD} to outsource the \acp{mvm} in the linear operations while computing nonlinear operations locally.
In this section, we present the implementation details and performance optimizations.

\parhead{Architecture of an LLM (Qwen3-32B as an example)}
An LLM generates its response \emph{autoregressively}, i.e., one token at a time, each conditioned on the prompt and on the tokens emitted so far.
Each token is computed by one \emph{forward pass} of the model.
A forward pass alternates between linear operations, namely \acp{mvm} with the weight matrices of the model, and nonlinear operations such as normalization, attention, and activation functions.
The \acp{mvm} involve large weight matrices and account for the bulk of the computation.
We denote the \acp{mvm} of a forward pass by $\by_j=\bW_j\bx_j$ for $j\in[T]$, in the order in which they are evaluated.
For Qwen3-32B~\cite{qwen3technicalreport}, the main model of our evaluation, $T=257$: each of the 64 transformer layers computes four products with $\bW_\text{qkv}$, $\bW_o$, $\bW_\text{gu}$, and $\bW_{\mathrm{down}}$, plus a final product with $\bW_{\mathrm{lm}}$.
The last \ac{mvm} yields the logits, from which the next token is selected.

Inference begins by converting the natural language prompt into tokens, and then into embedding vectors.
Next, in the \emph{prefill} stage (prompt processing), a single forward pass processes embedding vectors of all prompt tokens and emits the first token of the response (each $\bW_j$ multiplies a matrix with embedding vectors as columns).
Then, in the \emph{decode} stage (token generation), each further forward pass processes the last token, generating a new one.
Through attention, all earlier tokens affect the probability of the new token.
The decoding pass involves processing the {\em KV cache} of size proportional to the context (i.e., the number of prior tokens captured in attention),  which becomes more expensive as the context grows.

\subsubsection{Architecture of \protname}
\cref{fig:private-ml} illustrates the architecture of \protname on one forward pass of Qwen3-32B; the same flow applies to any model whose forward pass alternates \acp{mvm} with nonlinear operations.
The client evaluates the nonlinear operations locally and delegates every \ac{mvm} to the server using \ac{pvMVMD}: it derives each input $\bx_j$ from the outputs of the preceding \acp{mvm}, masks it, and unmasks the server's response to obtain $\by_j$.
The preprocessing of \ac{pvMVMD} is performed once for every weight matrix, independently of any prompt.

Processing a prompt gives rise to a large number of claims $\bW_j\bx_j=\by_j$.
The client has two options.
In \emph{real-time verification}, it verifies the claims of each token as soon as the token is generated, so that the user is shown only verified tokens.
In \emph{optimistic verification}, it verifies the claims of all tokens in one batch at the end of the response: the user sees tokens before they are verified, but the batch verification yields a mild throughput gain, as we show in \cref{subsubsec:eval-verification-only}.

In a bit more detail, the prefill stage can always be verified in batch, as all prompt tokens are processed in a single forward pass.
To do so, for each $j\in[T]$, the client maintains the running sums $\bX_j=\sum_i\rho_{j,i}\,\widehat{\bx}_{j,i}$ and $\bY_j=\sum_i\rho_{j,i}\,\widehat{\by}_{j,i}$ over the claims $\bW_j\widehat{\bx}_{j,i}=\widehat{\by}_{j,i}$ of all prompt tokens $i$, where $\rho_{j,i}$ is a secret random field element that the client draws afresh per matrix and token.
Optimistic verification uses the same technique, with running sums over the claims of all tokens of the response.
In both cases, the client verifies the single claim $\bW_j\bX_j=\bY_j$ for every $j$ using the verification of \cref{subsec:verification:construction}.
If any claim is false, the claim on the sums is also false except with probability $1/|\F|$.
To boost soundness, one can keep, e.g., two pairs of running sums with independent coefficients with soundness of $|\F|^{-2}\approx 2^{-62}$.

\begin{figure}[!t]
    \centering
    \resizebox{\linewidth}{!}{\begingroup
\definecolor{cl}{HTML}{2A78D6}\definecolor{clbg}{HTML}{E3EEFA}\definecolor{cllane}{HTML}{F6F9FD}
\definecolor{sv}{HTML}{C8551F}\definecolor{svbg}{HTML}{FADFD1}\definecolor{svlane}{HTML}{FDF7F3}
\definecolor{ink}{HTML}{1F2933}\definecolor{mut}{HTML}{6B7785}\definecolor{acc}{HTML}{1BAF7A}%
\begin{tikzpicture}[
  font=\sffamily\footnotesize\color{ink},
  op/.style={fill=clbg,inner sep=3pt,minimum width=3.6cm,minimum height=0.5cm,align=center},
  mm/.style={preaction={fill=svbg},pattern={Lines[angle=45,distance=2.8pt,line width=0.35pt]},pattern color=sv!50,inner sep=2.5pt,minimum width=2.7cm,minimum height=0.62cm,align=center},
  flow/.style={-{Latex[length=1.6mm]},draw=mut,line width=0.6pt},
  go/.style={-{Latex[length=1.6mm]},draw=sv,line width=0.6pt},
  lab/.style={font=\sffamily\scriptsize\color{sv},inner sep=1.5pt},
  note/.style={font=\sffamily\scriptsize\color{mut},align=left},
  tok/.style={draw=mut,line width=0.3pt,minimum size=0.24cm,inner sep=0pt}]
\def\xc{2.6}\def\xs{6.75}
\fill[cllane] (0,0.6) rectangle (4.85,-11.25);
\fill[svlane] (5.15,0.6) rectangle (8.35,-11.25);
\node at (2.4,0.3) {{\small\bfseries\color{cl}Client}\ \ {\scriptsize\color{cl}}};
\node at (\xs,0.3) {{\small\bfseries\color{sv}Server}\ \ {\scriptsize\color{sv}}};
\node[tok,fill=mut!30] (e0) at (1.10,-0.5) {};
\node[tok,fill=mut!30] (e1) at (1.40,-0.5) {};
\node[tok,fill=mut!30] (e2) at (1.70,-0.5) {};
\node[tok,fill=mut!30] (e3) at (2.00,-0.5) {};
\node[tok,fill=mut!30] (e4) at (2.30,-0.5) {};
\node[tok,fill=acc] at (2.60,-0.5) {};
\node[note,anchor=east,align=right] at (0.92,-0.5) {earlier\\[-1pt]tokens};
\node[note,anchor=west] at (2.80,-0.5) {last token};
\node[op] (emb) at (\xc,-1.30) {embedding lookup};
\node[note,anchor=west] at (2.70,-1.66) {one vector of 5,120 values};
\draw[mut,line width=0.5pt,dash pattern=on 2pt off 1.5pt] (0.6,-1.82) rectangle (8.25,-7.20);
\node[font=\sffamily\scriptsize\bfseries\color{mut},fill=cllane,inner sep=1.5pt,anchor=west] at (0.7,-1.82) {64 layers};
\node[op] (n1) at (\xc,-2.35) {RMSNorm};
\node[mm] (qkv) at (\xs,-2.35) {$\mathbf{W}_{\mathrm{qkv}}$\\[-2pt]{\scriptsize\color{ink!80}$10{,}240 \times 5{,}120$}};
\draw[go] (n1.east) -- node[above,lab] {mask} (qkv.west);
\node[op] (att) at (\xc,-3.45) {attention};
\node[mm] (o) at (\xs,-3.45) {$\mathbf{W}_o$\\[-2pt]{\scriptsize\color{ink!80}$5{,}120 \times 8{,}192$}};
\draw[go] (att.east) -- node[above,lab] {mask} (o.west);
\node[op] (n2) at (\xc,-4.55) {residual connection, RMSNorm};
\node[mm] (gu) at (\xs,-4.55) {$\mathbf{W}_{\mathrm{gu}}$\\[-2pt]{\scriptsize\color{ink!80}$51{,}200 \times 5{,}120$}};
\draw[go] (n2.east) -- node[above,lab] {mask} (gu.west);
\node[op] (act) at (\xc,-5.65) {SwiGLU};
\node[mm] (dn) at (\xs,-5.65) {$\mathbf{W}_{\mathrm{down}}$\\[-2pt]{\scriptsize\color{ink!80}$5{,}120 \times 25{,}600$}};
\draw[go] (act.east) -- node[above,lab] {mask} (dn.west);
\node[op] (res) at (\xc,-6.75) {residual connection};
\node[op] (fn) at (\xc,-7.80) {RMSNorm};
\node[mm] (lm) at (\xs,-7.80) {$\mathbf{W}_{\mathrm{lm}}$\\[-2pt]{\scriptsize\color{ink!80}$151{,}936 \times 5{,}120$}};
\draw[go] (fn.east) -- node[above,lab] {mask} (lm.west);
\node[op] (pick) at (\xc,-8.90) {select the next token};
\draw[go] (qkv.south) -- ++(0,-0.15) -| (att.north);
\node[lab,anchor=south west] at ($(att.north)+(0.06,0.06)$) {unmask};
\draw[go] (o.south) -- ++(0,-0.15) -| (n2.north);
\node[lab,anchor=south west] at ($(n2.north)+(0.06,0.06)$) {unmask};
\draw[go] (gu.south) -- ++(0,-0.15) -| (act.north);
\node[lab,anchor=south west] at ($(act.north)+(0.06,0.06)$) {unmask};
\draw[go] (dn.south) -- ++(0,-0.15) -| (res.north);
\node[lab,anchor=south west] at ($(res.north)+(0.06,0.06)$) {unmask};
\draw[go] (lm.south) -- ++(0,-0.15) -| (pick.north);
\node[lab,anchor=south west] at ($(pick.north)+(0.06,0.06)$) {unmask};
\draw[flow] (\xc,-0.66) -- (emb.north);
\draw[flow] (emb) -- (n1);
\draw[flow] (res) -- (fn);
\draw[flow,draw=acc,rounded corners=2pt] (e2.south) -- (1.70,-0.9) -- (0.2,-0.9) -- (0.2,-3.45) -- (att.west);
\node[note,text=acc!60!black,rotate=90,anchor=north] at (0.22,-2.18) {through the KV cache};
\node[draw=mut,line width=0.5pt,fill=white,inner sep=3pt,minimum width=2.2cm,minimum height=0.5cm] (dec) at (\xc,-9.85) {last token?};
\node[fill=ink!10,inner sep=3pt,minimum width=3.6cm,minimum height=0.5cm] (ver) at (\xc,-10.80) {optimistic verification};
\draw[flow] (pick) -- (dec);
\draw[flow] (dec) -- node[right,font=\sffamily\scriptsize\color{mut},inner sep=2pt] {yes} (ver);
\draw[flow,draw=acc,rounded corners=3pt,line width=0.8pt] (dec.west) -- node[above,pos=0.3,font=\sffamily\scriptsize\color{acc!60!black},inner sep=1.5pt] {no} (-0.25,-9.85) -- (-0.25,-0.08) -- (2.60,-0.08) -- (2.60,-0.35);
\node[note,text=acc!60!black,rotate=90,anchor=south] at (-0.3,-4.96) {the selected token is the input of the next pass};
\node[fill=ink!9,minimum size=0.24cm,inner sep=0pt] at (5.60,-10.80) {};
\node[note,anchor=west] at (5.77,-10.80) {nonlinear};
\node[preaction={fill=ink!9},pattern={Lines[angle=45,distance=2.8pt,line width=0.35pt]},pattern color=ink!60,minimum size=0.24cm,inner sep=0pt] at (7.35,-10.80) {};
\node[note,anchor=west] at (7.52,-10.80) {linear};
\end{tikzpicture}%
\endgroup}%
    \caption{%
    One decode pass of \protname on Qwen3-32B in the privacy mode.
    For clarity, it omits the quantization and rescaling steps before masking and after unmasking.}
    \label{fig:private-ml}
\end{figure}

\parhead{Boost mode}
Generating input masks and computing the corresponding output masks are a significant part of client overhead. 
Since these masks are independent of the client input, they can be generated beforehand (e.g., when the client's device is idle).
This optimization is similar to how MPC protocols generate Beaver Triples ahead of time to accelerate the online phase.
We call this the \emph{boost mode} of \protname.
We refer to the original workflow where the client generates these masks during inference as the \emph{standard mode}; we implement \textit{pipelining} so the masks for step $j+1$ are generated while waiting for the server to complete step $j$. 

\parhead{Verification-only mode}
Some users may not require input privacy, for example when their prompts contain no sensitive information or when they are willing to share them with the service provider, but may still want assurance that the provider performs the inference correctly.
In this case, the masking is unnecessary and, importantly, so is the interaction for every \ac{mvm}.
The client sends the prompt to the server.
The server evaluates the complete model, including the nonlinear operations, and streams the output $\by_j$ of every \ac{mvm}, together with every generated token, to the client.
The client replays the nonlinear operations on the received outputs, which yields the input $\bx_j$ of every \ac{mvm}, and verifies the claims $\bW_j\bx_j=\by_j$ as above, with the pairs $(\bx_j,\by_j)$ in the running sums.
\cref{fig:verifiable-ml} illustrates one decode pass with real-time verification.
It also checks that every token output by the server is the one with the largest logit. 
This removes two costs of the privacy mode.
First, the protocol is non-interactive after the prompt is sent: the server completes the inference without any input from the client, and the client replays and verifies concurrently, as the outputs arrive.
The network's round-trip time therefore does not affect the throughput: a request takes a single round trip.
Second, only the outputs $\by_j$ are transmitted, since the client recomputes the inputs $\bx_j$ itself.

The replay must reproduce the server's nonlinear operations exactly, although the two parties run on different hardware (e.g., a laptop CPU and a datacenter GPU), where floating-point results generally differ.
We therefore implement the nonlinear operations to be bitwise deterministic across platforms.
The IEEE~754 standard~\cite{ieee754} requires a small set of basic operations (addition, subtraction, multiplication, division, and square root) to return the same, correctly rounded result on every conforming platform.
We build every nonlinear operation from these operations alone, and we fix the order in which they are evaluated.
We avoid math libraries altogether, since their functions (e.g., exponentiation) carry no such guarantee.
\begin{figure}[t]
    \centering
    \resizebox{\linewidth}{!}{\begingroup
\definecolor{cl}{HTML}{2A78D6}\definecolor{clbg}{HTML}{E3EEFA}\definecolor{cllane}{HTML}{F6F9FD}
\definecolor{sv}{HTML}{C8551F}\definecolor{svbg}{HTML}{FADFD1}\definecolor{svlane}{HTML}{FDF7F3}
\definecolor{ink}{HTML}{1F2933}\definecolor{mut}{HTML}{6B7785}%
\begin{tikzpicture}[
  font=\sffamily\footnotesize\color{ink},
  op/.style={fill=clbg,inner sep=3pt,minimum width=3.1cm,minimum height=0.55cm,align=center},
  sop/.style={fill=svbg,inner sep=3pt,minimum width=3.1cm,minimum height=0.55cm,align=center},
  mm/.style={preaction={fill=svbg},pattern={Lines[angle=45,distance=2.8pt,line width=0.35pt]},pattern color=sv!50,inner sep=3pt,minimum width=3.1cm,minimum height=0.55cm,align=center},
  flow/.style={-{Latex[length=1.6mm]},draw=mut,line width=0.6pt},
  go/.style={-{Latex[length=1.6mm]},draw=sv,line width=0.6pt},
  lab/.style={font=\sffamily\scriptsize,inner sep=1.5pt},
  note/.style={font=\sffamily\scriptsize\color{mut}}]
\fill[cllane] (0,0.55) rectangle (3.9,-4.75);
\fill[svlane] (4.3,0.55) rectangle (8.2,-4.75);
\node at (1.95,0.25) {\small\bfseries\color{cl}Client};
\node at (6.25,0.25) {\small\bfseries\color{sv}Server};
\node[op] (prompt) at (1.95,-0.45) {send the prompt};
\draw[mut,line width=0.5pt,dash pattern=on 2pt off 1.5pt] (0.1,-0.95) rectangle (8.1,-3.25);
\node[note,anchor=south west] at (0.15,-3.22) {repeated for every matrix $\mathbf{W}$ of every layer};
\node[sop] (f)  at (6.25,-1.35) {nonlinear operation $f$};
\node[mm]  (w)  at (6.25,-2.4)  {matrix $\mathbf{W}$};
\node[op]  (cf) at (1.95,-1.35) {replay $f$};
\node[op]  (v)  at (1.95,-2.4)  {verify $\mathbf{W}x=y$\\[-1pt]{\scriptsize(real-time verification)}};
\draw[flow] (prompt.east) -| node[above,pos=0.25,lab,text=mut] {prompt} (f.north);
\draw[go] (f) -- (w);
\draw[flow] (cf) -- node[right,lab,text=mut,xshift=2pt] {$x$} (v);
\draw[go] (w.west) -- node[above,lab,text=sv] {$y$} (v.east);
\node[sop] (tok) at (6.25,-3.95) {select the next token};
\node[op]  (chk) at (1.95,-3.95) {check the token\\[-1pt]{\scriptsize(largest logit)}};
\draw[go] (w) -- (tok);
\draw[go] (tok.west) -- node[above,lab,text=sv] {token} (chk.east);
\end{tikzpicture}%
\endgroup}%
    \caption{%
    One forward pass of \protname in the verification-only mode.}
    \label{fig:verifiable-ml}
\end{figure}

\subsubsection{Model quantization}
The protocols of \cref{sec:verification,sec:delegation} operate over a finite field, whereas the weights and activations of a model are floating-point numbers.
We therefore need to quantize the weights and the inputs of every \ac{mvm} into integers, which we embed in the BabyBear field $\F_q$, with $q=2^{31}-2^{27}+1$.
Quantizing a value means dividing it by a \emph{step} and rounding the result to the nearest integer.
The choice of the step matters.
With a step that is too big, many different values round to the same integer, and this loss of precision degrades the model's accuracy.
With a step that is too small, the resulting integers are large in magnitude, and the sum computed by an \ac{mvm} can exceed the modulus $q$ and wrap around, in which case the result is wrong.
We choose the steps such that a wrap-around cannot occur for any input, while the loss in accuracy remains negligible (see \cref{subsubsec:eval-accuracy}).
We describe the quantization of a single \ac{mvm} $\by=\bW\bx$.

\parhead{Weights}
Each row of a weight matrix is quantized with its own step.
Since the magnitudes of the weights vary considerably across rows, a per-row step yields a finer quantization than a per-matrix step, at the moderate cost of storing one step per row.
For row $i$ of $\bW$, we set the step to $s_i=\max_k|w_{ik}|/127$ and round the entries to the 8-bit integers $w^{\mathrm{int}}_{ik}=\mathrm{round}(w_{ik}/s_i)$.
The weights are quantized once, in advance; we write $\bW^{\mathrm{int}}$ for the resulting integer matrix.

\parhead{Input vectors}
An input vector is quantized with a step that depends on the weight matrix applied to it.
Our algorithm works as follows.
For each matrix $\bW$, we define $L_{\bW}=\max_i\sum_k|w^{\mathrm{int}}_{ik}|$ to be the largest row sum of the quantized weights, and fix a bound $A_{\bW}$ such that $A_{\bW}\cdot L_{\bW}\le(q-1)/2$.
An input $\bx$ of $\bW$ is then quantized with the step $s_{\bx}=\max_k|x_k|/A_{\bW}$, which maps its entries to integers $x^{\mathrm{int}}_k=\mathrm{round}(x_k/s_{\bx})\in[-A_{\bW},A_{\bW}]$; we write $\bx^{\mathrm{int}}$ for the resulting integer vector.
Since $A_{\bW}\cdot L_{\bW}\le(q-1)/2$, the product $\bW^{\mathrm{int}}\bx^{\mathrm{int}}$ cannot wrap around.

\parhead{Rescaling the result}
The server returns the integer vector $\by^{\mathrm{int}}=\bW^{\mathrm{int}}\bx^{\mathrm{int}}$.
The client recovers the floating-point result by multiplying every entry by the two steps, $y_i\approx s_i\,s_{\bx}\,y^{\mathrm{int}}_i$, and evaluates the next nonlinear operation on it in floating point.
The recovered value is close to the product $\sum_k w_{ik}x_k$ of the original model but not equal to it: the integer product $\by^{\mathrm{int}}$ is exact, but $s_i\,w^{\mathrm{int}}_{ik}$ and $s_{\bx}\,x^{\mathrm{int}}_k$ only approximate $w_{ik}$ and $x_k$, since both were rounded.
Because the client evaluates the nonlinear operations, it returns to floating point after every \ac{mvm} and quantizes the next input from scratch.
Quantization thus affects one \ac{mvm} at a time.
In a system that evaluates the entire forward pass in the field, by contrast, the output of one \ac{mvm} is the input of the next, so the bound on the integers is multiplied by $L_{\bW}$ at every \ac{mvm}, which tends to require a much larger field to avoid a wrap-around.
In \protname, this growth does not occur, because the input of every \ac{mvm} is quantized from scratch; a 31-bit field is therefore large enough.

\subsubsection{Implementation and optimizations}

We implement \protname in about 38,000 lines of Rust and 2,000 lines of CUDA, with support for Qwen3, Qwen2.5, and Llama~4 model families.
Our implementation includes optimizations to reduce the number of round trips in the privacy mode, compresses the server's outputs in the verification-only mode, and to optimally leverage a small GPU on the client that is too small to hold the entire model.

\parhead{Reducing round complexity}
In the privacy mode, every \ac{mvm} costs one round of interaction, and the rounds are sequential, since the input of an \ac{mvm} depends on the outputs of the preceding ones.
We reduce the number of rounds by exploiting that several weight matrices of a transformer layer are applied to the same input: the query, key, and value projections of the attention, and the gate and up projections of the feed-forward network.
We stack the matrices of each group vertically into one matrix, $\bW_\text{qkv}$ and $\bW_\text{gu}$ respectively, and delegate it as a single \ac{mvm}.
A layer then takes 4 rounds instead of 7, and a forward pass of Qwen3-32B takes 257 rounds instead of 449.

\parhead{Reducing communication}
In the verification-only mode, the entries of the streamed vectors $\by^{\mathrm{int}}$ are typically far smaller in magnitude than the modulus, so the server stores each block of 32 entries at the bit width of its largest entry; an entry then takes about 2.9 bytes on average instead of 4.
(In the privacy mode, the server's outputs are masked, hence uniformly distributed, and cannot be compressed.)

\parhead{GPU-equipped clients}
\protname targets computationally weak clients, including those with a GPU that is too small to hold the model and thus cannot run it efficiently on their own.
Nonetheless, such a client can hold some of the weight matrices and evaluate their \acp{mvm} itself, which saves communication and verification.
The question is how to split the work optimally between the client and the server under the client's memory budget.
Our evaluation shows that GPU clients are bound by the network, so we use a simple heuristic that maximizes the reduction in bandwidth.
Storing a matrix with $m$ rows and $n$ columns on the client's GPU costs $O(mn)$ memory and saves $m$ values per token in bandwidth, so the bandwidth saved per byte of memory is proportional to $1/n$.
The client therefore takes the matrices in increasing order of $n$, i.e., the narrowest matrices first (for Qwen3-32B, $\bW_\text{qkv}$ and $\bW_\text{gu}$, then $\bW_o$, then $\bW_{\mathrm{down}}$), until its budget is exhausted.
Looking ahead to \cref{subsubsec:eval-verification-only}, on a desktop with a 16\,GB GPU, the client holds 71 of the 256 per-layer matrices of Qwen3-32B in 11.1\,GB, which halves the communication from 13.8 to 7.1\,MB per generated token.

Another way to use the client's GPU is to perform the prefill stage on the client.
If the client can afford to keep a copy of the quantized weights on disk, it can evaluate the entire prefill pass efficiently, because prefill processes all prompt tokens in one batch, with matrix--matrix products instead of \acp{mvm}.
The weight matrices are loaded from disk once and applied to all prompt tokens, so the cost of loading them is paid once per prompt.
In the decode stage, by contrast, every token would require loading all the weights again, which is much slower than delegating.

\subsection{Evaluation}
\label{subsec:evaluation-gpu}

We evaluate \protname in three parts: the verification-only mode on weak clients, against local inference with llama.cpp (\cref{subsubsec:eval-verification-only}); the privacy mode, against MOSAIC~\cite{chiang2026mosaic} (\cref{subsubsec:eval-privacy}); and the accuracy cost of quantization (\cref{subsubsec:eval-accuracy}).

\subsubsection{Verification-only mode}
\label{subsubsec:eval-verification-only}

To demonstrate that a computationally weak client can verify the inference of large models by an untrusted server, we evaluate \protname on two such clients: a CPU-only client and a client with a desktop GPU (called the \textit{GPU client}).
For both clients we evaluate two models: Qwen3-32B~\cite{qwen3-32b-fp8} and Llama~4~\cite{llama4-maverick-fp8}, both with 8-bit weights.
For Qwen3-32B, which the clients can barely run locally, \protname generates tokens 81--270$\times$ faster than llama.cpp.
For Llama~4, which they cannot run at all, the CPU and GPU clients generate 7--15 and 24--33 tokens/s.

\parhead{Setup}
Our CPU client is a Lenovo ThinkPad X9-15 Gen~1 laptop with an Intel Core Ultra~7 268V processor (8 cores), 30\,GB of memory and no discrete GPU.
The GPU client is a desktop with an RTX~4080 (16\,GB of GPU memory), limited to 16\,GB of memory and 4 cores, so that its total memory, 32\,GB, is about that of the CPU client.

The servers run in data centers.
For Qwen3-32B, we use an A100 SXM4 with 80\,GB of GPU memory, rented from a commercial cloud provider and reached over a WAN; we measure a round-trip time of 10\,ms to the client and a bandwidth of 1\,Gb/s.
For Llama~4, we use four H200s in a university research cluster, reached over a WAN with a round-trip time of 8--9\,ms and a bandwidth of 1\,Gb/s.
Unless otherwise specified, every test is run once to warm up and three times measured, and we report the mean.

\parhead{Baselines}
The baseline is llama.cpp~\cite{llamacpp} on the same client.
llama.cpp is a widely used open-source engine for local inference, optimized with SIMD kernels and batched prompt processing.
We run it unmodified (commit \texttt{669204e} on the CPU client and \texttt{7071ba3}, built with CUDA, on the GPU client) with its 8-bit format \texttt{Q8\_0}.
For llama.cpp, we run each generation test once, and we compute the time to first token from its prompt-processing throughput.
We have no baseline for Llama~4, which is far too large to run on either client.

\parhead{Benchmarks and metrics}
To compare \protname with baselines, we benchmark token generation with varying context lengths.
We use the convention of \texttt{llama-bench} to denote benchmark runs: \texttt{tgM} produces M tokens one at a time after a one-token prompt; \texttt{ppN+tgM} processes an N-token prompt and then generates M tokens.
We use \texttt{tg128} and \texttt{ppN+tg128} for N = 512, 1,024, 2,048 and 4,096.
In all tests, decoding is greedy: the next token is the one with the largest logit; alternatively, the next token can be sampled pseudorandomly using a preshared seed wit no additional overhead.

For each test, we report the two quantities that a user experiences.
The \textbf{time to first token} (TTFT) is the time from sending the prompt until the first token of the response is available and verified.
The \textbf{generation throughput} (tokens/s) is the number of generated tokens over the time to generate them, after the first token.
All times are online times and exclude one-time setup.

\begin{figure*}[t]
\centering
\def\vodata{data/2.csv}%
\definecolor{mvblue}{HTML}{2A78D6}\definecolor{mvorange}{HTML}{EB6834}%
\pgfplotsset{
  vo mav/.style={ybar, bar width=5pt, bar shift=-2.9pt, fill=mvblue, draw=mvblue, line width=0.3pt},
  vo llama/.style={ybar, bar width=5pt, bar shift=2.9pt, fill=mvorange, draw=mvorange, line width=0.3pt},
  vo x/.style={xmin=-0.6, xmax=4.6, xtick={0,...,4}, xticklabels={0,512,1k,2k,4k}, xlabel={prompt length (tokens)}},
  vo speedup/.style={only marks, mark=none, forget plot, point meta=explicit symbolic,
                     nodes near coords={\pgfplotspointmeta\rlap{$\times$}},
                     every node near coord/.append style={font=\tiny, text=black, anchor=south, inner sep=1.5pt, xshift=#1}},
  vo value/.style={only marks, mark=none, forget plot, point meta=y, nodes near coords,
                   every node near coord/.append style={font=\tiny, text=black!75, anchor=south, inner sep=1pt, xshift=#1,
                                                        /pgf/number format/.cd, fixed, precision=2}},
}%
\begin{tikzpicture}
\begin{groupplot}[
  group style={group size=4 by 1, horizontal sep=1.25cm},
  scale only axis, width=2.92cm, height=2.6cm,
  xticklabel shift=-3pt, xlabel style={font=\footnotesize, yshift=2pt},
  ymin=0, axis lines*=left, tick align=outside, xtick style={draw=none},
  ymajorgrids, grid style={black!12}, tick label style={font=\footnotesize},
  title style={font=\footnotesize, at={(0.5,1.0)}, anchor=south, inner sep=0pt, text height=1.5ex, text depth=0.3ex},
  ylabel style={font=\footnotesize, yshift=-1pt},
  legend style={font=\footnotesize, draw=none, fill=none, legend columns=2, at={(2.64,1.2)}, anchor=south,
                /tikz/every even column/.append style={column sep=6pt}},
  area legend,
]
\nextgroupplot[title={Generation}, ylabel={tokens/s}, vo x, ymax=11]
\addplot[vo mav] table[col sep=comma, x=pos, y=cpu_mav_rt_gen_tps]{\vodata};
\addplot[vo llama] table[col sep=comma, x=pos, y=cpu_llama_gen_tps]{\vodata};
\addplot[vo speedup=-2.9pt] table[col sep=comma, x=pos, y=cpu_mav_rt_gen_tps, meta=cpu_gen_x]{\vodata};
\addplot[vo value=4.2pt, every node near coord/.append style={/pgf/number format/fixed zerofill}] table[col sep=comma, x=pos, y=cpu_llama_gen_tps]{\vodata};
\legend{\protname{} (ours), llama.cpp (baseline)}
\nextgroupplot[title={TTFT}, ylabel={seconds}, vo x, xmin=0.4, xtick={1,...,4}, xticklabels={512,1k,2k,4k}, ymax=1400]
\addplot[vo mav] table[col sep=comma, x=pos, y=cpu_mav_rt_ttft_s]{\vodata};
\addplot[vo llama] table[col sep=comma, x=pos, y=cpu_llama_ttft_s]{\vodata};
\addplot[vo speedup=2.9pt] table[col sep=comma, x=pos, y=cpu_llama_ttft_s, meta=cpu_ttft_x]{\vodata};
\nextgroupplot[xshift=0.55cm, title={Generation}, ylabel={tokens/s}, vo x, ymax=23]
\addplot[vo mav] table[col sep=comma, x=pos, y=gpu_mav_rt_gen_tps]{\vodata};
\addplot[vo llama] table[col sep=comma, x=pos, y=gpu_llama_gen_tps]{\vodata};
\addplot[vo speedup=-2.9pt] table[col sep=comma, x=pos, y=gpu_mav_rt_gen_tps, meta=gpu_gen_x]{\vodata};
\addplot[vo value=4.2pt, every node near coord/.append style={/pgf/number format/fixed zerofill}] table[col sep=comma, x=pos, y=gpu_llama_gen_tps]{\vodata};
\nextgroupplot[title={TTFT}, ylabel={seconds}, vo x, xmin=0.4, xtick={1,...,4}, xticklabels={512,1k,2k,4k}, ymax=85]
\addplot[vo mav] table[col sep=comma, x=pos, y=gpu_mav_rt_ttft_s]{\vodata};
\addplot[vo llama] table[col sep=comma, x=pos, y=gpu_llama_ttft_s]{\vodata};
\addplot[vo speedup=2.9pt] table[col sep=comma, x=pos, y=gpu_llama_ttft_s, meta=gpu_ttft_x]{\vodata};
\end{groupplot}
\node[font=\footnotesize, anchor=north] at ([xshift=0.625cm, yshift=-0.95cm]group c1r1.south east) {(a) CPU client};
\node[font=\footnotesize, anchor=north] at ([xshift=0.625cm, yshift=-0.95cm]group c3r1.south east) {(b) GPU client};
\end{tikzpicture}%
%
\caption{Qwen3-32B in the verification-only mode: generation throughput and TTFT of \protname and llama.cpp on both clients, for \texttt{tg128} (prompt length 0) and \texttt{ppN+tg128}. The numbers above the bars are the speedups of \protname.}\label{fig:eval-vo}
\end{figure*}

\parhead{Results for Qwen3-32B}
\cref{fig:eval-vo} presents the evaluation results for Qwen3-32B.
Run locally, the model is not usable on either client: because the model does not fit in memory, each forward pass reads part of the model from disk, which drastically limits the performance.
For instance, a 128-token response takes close to an hour on the CPU client and about ten minutes on the GPU client.
With \protname, the same clients generate 8 and 18 tokens/s, a speedup of more than 200$\times$ and of 81--110$\times$, respectively.

The evaluation shows that \protname is network-bound for short context lengths.
For the CPU client, the server sends the output of all matrix--vector products, about 13.8\,MB per generated token, which takes 0.12\,s on the client's 1\,Gb/s link.
The GPU client evaluates 71 of the 256 per-layer matrices itself (\cref{subsec:implementation-ml}); this halves the traffic and, accordingly, doubles the throughput.
Being network-bound also explains why the throughput is almost independent of the prompt length: the client's computation becomes the bottleneck only on the CPU client, and only after prompts of 2,048 tokens or more.

The TTFT behaves differently on the two clients.
On the CPU client, the TTFT is similarly network-bound: the outputs for all prompt tokens must be transmitted to the client.
Even so, \protname still produces the first token about 2$\times$ sooner than llama.cpp.
The GPU client can avoid this transfer by evaluating the prefill pass itself (\cref{subsec:implementation-ml}), much faster than llama.cpp.
The reason is that llama.cpp places as many layers as fit on the GPU, in order, and evaluates the remaining layers on the CPU, whereas \protname's client needs only one layer on the GPU at a time for the prompt.

\textit{Optimistic verification.}
The client can also defer the verification to the end of the response (\cref{subsec:implementation-ml}).
We repeat the Qwen3-32B experiments in this setting.
The difference is small, because a verification is cheap: it takes 8.5\,ms on the CPU client and 4.8\,ms on the GPU client.
The TTFT is therefore practically unchanged, although the first token is then displayed before it is verified.
The throughput is also unchanged when generation is network-bound, since the client's verification time is hidden by the transmission time.
It improves only when the CPU client's computation is the bottleneck, after prompts of 2,048 and 4,096 tokens, and then by 6--7\%.
The rest of the evaluation therefore reports real-time verification unless stated otherwise.

\begin{figure}[t]
\centering
\def\lfdata{data/3.csv}%
\definecolor{mvblue}{HTML}{2A78D6}\definecolor{mvgreen}{HTML}{1BAF7A}%
\pgfplotsset{
  lf cpu/.style={mvblue, line width=1pt, mark=*, mark size=1.6pt, mark options={fill=mvblue, draw=white, line width=0.4pt}, unbounded coords=jump},
  lf gpu/.style={mvgreen, line width=1pt, mark=square*, mark size=1.6pt, mark options={fill=mvgreen, draw=white, line width=0.4pt}, unbounded coords=jump},
  lf cpu opt/.style={lf cpu, densely dashed, mark=o, mark options={draw=mvblue, fill=white, line width=0.6pt}},
  lf gpu opt/.style={lf gpu, densely dashed, mark=square, mark options={draw=mvgreen, fill=white, line width=0.6pt}},
  lf x/.style={xmin=-0.3, xmax=4.3, xtick={0,...,4}, xticklabels={0,512,1k,2k,4k}, xlabel={prompt length (tokens)}},
}%
\begin{tikzpicture}
\begin{groupplot}[
  group style={group size=2 by 1, horizontal sep=1.25cm},
  scale only axis, width=3.05cm, height=2.6cm,
  xticklabel shift=-3pt, xlabel style={font=\footnotesize, yshift=2pt},
  ymin=0, axis lines*=left, tick align=outside, 
  ymajorgrids, grid style={black!12}, tick label style={font=\footnotesize},
  title style={font=\footnotesize, at={(0.5,1.0)}, anchor=south, inner sep=0pt, text height=1.5ex, text depth=0.3ex},
  ylabel style={font=\footnotesize, yshift=-1pt},
  legend cell align=left,
  legend style={font=\footnotesize, draw=none, fill=none, legend columns=2, at={(1.2,1.2)}, anchor=south,
                /tikz/every even column/.append style={column sep=6pt}},
  legend image post style={line width=1pt},
]
\nextgroupplot[title={Generation}, ylabel={tokens/s}, lf x, ymax=40]
\addplot[lf cpu] table[col sep=comma, x=pos, y=cpu_gen_tps]{\lfdata};
\addplot[lf gpu] table[col sep=comma, x=pos, y=gpu_gen_tps]{\lfdata};
\addplot[lf cpu opt] table[col sep=comma, x=pos, y=cpu_opt_gen_tps]{\lfdata};
\addplot[lf gpu opt] table[col sep=comma, x=pos, y=gpu_opt_gen_tps]{\lfdata};
\legend{CPU client (real-time ver.), GPU client (real-time ver.), CPU client (opt.\ ver.), GPU client (opt.\ ver.)}
\nextgroupplot[title={TTFT}, ylabel={seconds}, lf x, xmin=0.7, xtick={1,...,4}, xticklabels={512,1k,2k,4k}, ymax=320]
\addplot[lf cpu] table[col sep=comma, x=pos, y=cpu_ttft_s]{\lfdata};
\addplot[lf gpu] table[col sep=comma, x=pos, y=gpu_ttft_s]{\lfdata};
\end{groupplot}
\end{tikzpicture}%
%
\caption{Llama~4 in the verification-only mode: generation throughput and TTFT of \protname on both clients, for \texttt{tg128} (prompt length 0) and \texttt{ppN+tg128}. There is no local baseline.
Real-time ver.\ = real-time verification; opt.\ ver.\ = optimistic verification (\cref{subsec:implementation-ml}).}\label{fig:eval-llama4}
\end{figure}

\parhead{Results for Llama~4}
Llama~4 cannot run on either client at all, since its weights take about 400\,GB, so there is no local baseline.
With \protname, both clients can nevertheless use it (\cref{fig:eval-llama4}): across the prompt lengths, the CPU client generates 7--15 tokens/s and the GPU client 24--33 tokens/s, which is fast enough for interactive use.
The GPU client is again at least twice as fast, because it evaluates part of the model itself.
Interestingly, both clients are faster than with Qwen3-32B, a model 12$\times$ smaller, because Llama~4 is a mixture of experts~\cite{shazeer2017moe}, in which each token uses only a fraction of the parameters, so the server sends less data per token to the client.
The throughput decreases with the prompt length as the client's computation becomes the bottleneck.
Unlike with Qwen3-32B, it is infeasible for the GPU client to evaluate the prefill pass itself, so the TTFT is the same on both clients.
With optimistic verification, the GPU client generates up to 8\% faster when client is computation-bound.

For comparison, FlashProve~\cite{flashprove}, the state-of-the-art zero-knowledge proof system for LLM inference, evaluated on the same model, reports 58.8\,s on an H100 to prove one next token after a 1,024-token prompt, and about 21\,s to verify it.
With \protname, the server does no work beyond the inference, and the throughput above includes the client's verification.

\parhead{Cost of refreshing the challenges}
\label{subsubsec:eval-slalom}
Recall that the client verifies a claim $\by=\bM\bx$ with a secret sparse challenge $\be$, by checking $(\be^\top\bQ)\bx=\be^\top\bG^\top\by$, where $\bQ=\bG^\top\bM$ is computed in the preprocessing (\cref{subsec:verification:construction}).
In the experiments above, the client computes the projections $\be^\top\bQ$ ahead of time and reuses them for all requests.
The client may need to refresh its challenges when it detects server cheating, or, as a layer of defense in depth, may want to refresh them proactively after every request.
We measure the overhead of proactive refreshing.

The baseline of comparison is the standard way to obtain a fresh challenge: sample a uniform vector $\bc$ and compute its projection $\bc^\top\bM$ for all $\bM$.
This is used in, e.g., Slalom~\cite{slalom}.

In our experiment, we repeat the Qwen3-32B experiment on the GPU client with refreshing after each request.
With both methods, the client computes new challenges in the background, in parallel to the current request;
in optimistic verification, this is to compute the challenge for the current request; in real-time verification, this is to compute the challenge for the {\em next} request.
If the refresh takes longer than the current request, the next request must wait for it.
The two are identical in performance, so we present optimistic verification.
We measure the TTFT and the throughput that the client can sustain over consecutive requests.
The refresh takes a fixed time per request, so it matters most for short requests.
We therefore add four short tests (\texttt{pp16+tg32}, \texttt{pp32+tg64}, \texttt{pp128+tg64} and \texttt{pp1024+tg64}) to \texttt{pp512+tg128}.

\begin{figure}[t]
\centering
\resizebox{\linewidth}{!}{\def\sldata{data/4.csv}%
\definecolor{mvblue}{HTML}{2A78D6}\definecolor{mvorange}{HTML}{EB6834}\definecolor{mvgreen}{HTML}{1BAF7A}%
\pgfplotsset{
  sl/.style={ybar, bar width=5pt, line width=0.3pt},
  sl wait/.style={ybar, bar width=5pt, line width=0.3pt, preaction={fill=white},
                  pattern={Lines[angle=45,distance=1.6pt,line width=0.35pt]}},
}%
\begin{tikzpicture}
\begin{groupplot}[
  group style={group size=2 by 1, horizontal sep=1.15cm},
  scale only axis, width=3.7cm, height=2.5cm,
  xmin=-0.5, xmax=4.5, xtick={0,...,4}, xticklabels={pp16+tg32, pp32+tg64, pp128+tg64, pp1024+tg64, pp512+tg128},
  xticklabel shift=-3pt, ymin=0, x tick label style={font=\tiny, rotate=35, anchor=north east, inner sep=1pt},
  axis lines*=left, tick align=outside, xtick style={draw=none},
  ymajorgrids, grid style={black!12}, tick label style={font=\scriptsize},
  title style={font=\scriptsize, at={(0.5,1.0)}, anchor=south, inner sep=0pt, text height=1.5ex, text depth=0.3ex},
  ylabel style={font=\scriptsize, yshift=-1pt},
  legend style={font=\scriptsize, draw=none, fill=none, legend columns=3, at={(1.15,1.22)}, anchor=south,
                /tikz/every even column/.append style={column sep=5pt}},
  area legend,
]
\nextgroupplot[title={Sustained throughput}, ylabel={tokens/s}, ymax=20]
\addplot[sl, bar shift=-2.8pt, fill=mvblue, draw=mvblue] table[col sep=comma, x expr=\coordindex, y=mav_steady]{\sldata};
\addplot[sl, bar shift=2.8pt, fill=mvorange, draw=mvorange] table[col sep=comma, x expr=\coordindex, y=slalom_gpu_steady]{\sldata};
\addlegendimage{sl wait, draw=black!60, pattern color=black!60}
\legend{\protname{} (ours), standard, waiting for challenges}
\nextgroupplot[title={TTFT}, ylabel={seconds}, ymax=7]
\addplot[sl, bar shift=-2.8pt, fill=mvblue, draw=mvblue] table[col sep=comma, x expr=\coordindex, y expr=\thisrow{mav_wait}+\thisrow{mav_ttft}]{\sldata};
\addplot[sl, bar shift=2.8pt, fill=mvorange, draw=mvorange] table[col sep=comma, x expr=\coordindex, y expr=\thisrow{slalom_gpu_wait}+\thisrow{slalom_gpu_ttft}]{\sldata};
\addplot[sl wait, bar shift=2.8pt, draw=mvorange, pattern color=mvorange] table[col sep=comma, x expr=\coordindex, y=slalom_gpu_wait]{\sldata};
\end{groupplot}
\end{tikzpicture}%
}%
\caption{Qwen3-32B on the GPU client with fresh challenges for every request: sustained throughput over consecutive requests, and TTFT. The hatched part of the TTFT is the time that a request waits for its challenges; the rest is the prompt pass, which the standard method slows down by refreshing the next challenges at the same time.}\label{fig:eval-slalom}
\end{figure}

\cref{fig:eval-slalom} reports the results.
\protname completes a refresh in about 1.5\,s, less than the shortest request takes.
Refreshing is therefore free.
The standard method needs more than 7\,s per refresh, dominated by the time to read the 32\,GB of weights from the SSD.
Requests that are shorter than this must wait.
This is because the standard method involves reading the entire model and, for large models that cannot fit in memory, reading it from disk.
In \protname, computing $\be^\top\bQ$ only involves reading a few rows of $\bQ$.
$\bQ$ is larger, so this is a storage--performance tradeoff.
In our experiment, using a Reed--Solomon code of rate 1/2 (with 40 bits of soundness), $\bQ$ takes 8 bytes per weight (256\,GB for Qwen3-32B), and a refresh reads 747\,MB of it (the model weights required by the standard method are 32\,GB).

On the four short tests, the sustained throughput of \protname is 1.5--3$\times$ higher.
For longer requests, such as \texttt{pp512+tg128}, the two methods reach the same throughput.
The standard method also worsens the TTFT, for two reasons.
A short request first waits for its challenges (the hatched part in \cref{fig:eval-slalom}).
Its prompt pass then takes about twice as long, because the refresh for the next request runs at the same time and competes with it for the SSD.
Together, the TTFT of a short request is 3--5$\times$ longer than with \protname.

\parhead{Storage and memory usage}
When challenges are reused, the client's storage is small, because it can compute the challenge projections $\be^\top\bQ$ (one per weight matrix) ahead of time and then delete the model weights.
For Qwen3-32B, the CPU client stores 1.65\,GB: 84\,MB of challenge projections and quantization steps, and 1.57\,GB for the embedding table.
For Llama~4, whose weights take about 400\,GB, the client stores only 3.3\,GB: 1.2\,GB of challenge projections and quantization steps, and 2.1\,GB for the embedding table.
The memory usage is moderate, peaking at 6.5\,GB and 7.0\,GB in the Qwen3-32B and Llama~4 experiments, respectively.

To leverage its GPU, the GPU client stores the part of the model that it evaluates itself.
In our experiment, for Qwen3-32B, it keeps a copy of all 8-bit weights (31\,GB) for the prefill pass;
for Llama~4, it stores 28\,GB for the non-expert weights and 11\,GB of setup data. (Experts are run on the server.)
To support refreshing the challenges, a client needs to store $\bQ$ on disk as discussed above.

\subsubsection{Privacy mode}
\label{subsubsec:eval-privacy}

The privacy mode involves many round trips between the client and the server, so it is only feasible when the two are connected by a low-latency link.
It is most similar to the recent system MOSAIC~\cite{chiang2026mosaic} (CCS'26), in which the client and the servers reside in the same data center (connected by a low-latency link such as NVLink) but the client is trusted (e.g., it runs in a TEE) while the servers are not.
In this setting, trusted compute is the scarce resource, and the goal of outsourcing is for a trusted GPU to increase its throughput with the help of \textit{untrusted} accelerators~\cite{chiang2026mosaic}.
The relevant cost is therefore the time that the trusted side spends per token, in addition to the end-to-end throughput.

MOSAIC is the natural point of comparison: it targets the same deployment and, like \protname, outsources the linear operations of the model to untrusted GPUs.
The two systems make different trade-offs in what they guarantee: \protname verifies the server's results and keeps the outputs exact, whereas MOSAIC also hides the model from the server, but it does not verify the results and adds noise to the outputs.
The comparison therefore shows how efficiently each design delivers its guarantees in the same setting.

\parhead{Setting}
We run both systems on the hardware MOSAIC targets, H200 GPUs of one machine connected by NVLink, and on Qwen2.5-72B, a model of the same size (72B) as the one its paper evaluates.
We run MOSAIC's code~\cite{mosaic-code} with its own benchmark, in its faster mode, with one trusted and three untrusted H200s (three GPUs are required because its masked weights take about 290\,GB).
\protname runs in the privacy mode with optimistic verification, with the client on one H200 and the server on a second.
It uses the boost mode, in which the client prepares its one-time masks ahead of time, while MOSAIC generates its masks within the request.
Each system uses its own benchmark's prompts: 1,024, 4,096 and 7,168 tokens for \protname, and for MOSAIC a system prompt of 1,456, 4,519 or 6,703 tokens followed by a question of 45--67 tokens, processed in two passes.
We refer to these as 1k, 4k and 7k of context.
We report online time.

\begin{table}[t]
\centering
\resizebox{\columnwidth}{!}{\begin{tabular}{lrr}
\toprule
& \protname & MOSAIC \\
\midrule
GPUs (trusted + untrusted) & 1 + 1 & 1 + 3 \\
Prompt processing at 1k / 4k / 7k (t/s) & 1,010 / 895 / 645 & 211 / 261 / 247 \\
Generation at 1k / 4k / 7k context (t/s) & 21.6 / 18.9 / 18.4 & 1.84 / 1.83 / 1.83 \\
Trusted side per generated token (ms) & 19.1 / 25.7 / 27.0 & 42.8 / 43.1 / 42.8 \\
Verifiability & yes & no \\
Model privacy & no & yes \\
Outputs exact & yes & no \\
\bottomrule
\end{tabular}}
\caption{Privacy mode: \protname and MOSAIC, at about 1k, 4k and 7k tokens of context.}\label{tab:eval-privacy}
\end{table}

\parhead{Results}
\cref{tab:eval-privacy} reports the results.
\protname processes prompts 2.6--4.8$\times$ faster than MOSAIC and generates 10--12$\times$ faster, while using one untrusted GPU instead of three.
The gap in generation comes from the untrusted side: MOSAIC's untrusted GPUs take 479\,ms per generated token, whereas \protname's server takes 22\,ms.
The trusted side, the scarce resource in this setting, also does less work with \protname: 19--27\,ms per generated token, against 43\,ms with MOSAIC.
Recall, however, that the two systems offer different guarantees: \protname provides verifiability and does not introduce noise into its outputs, whereas MOSAIC hides the model.
Moreover, \protname runs in the boost mode: its times are online times and exclude the preparation of its masks, whereas MOSAIC's times include its mask generation.

\subsubsection{Accuracy}
\label{subsubsec:eval-accuracy}

\protname evaluates a model on quantized weights and inputs (\cref{subsec:implementation-ml}), so we measure how much accuracy the quantization costs.
By construction, no result can wrap around the modulus, so rounding is the only source of error.

\parhead{Methodology}
We run each model twice on the WikiText-2 test set~\cite{merity2017pointer}, the standard corpus for perplexity comparisons of quantized models.
The reference run uses released weights without any further quantization: the official FP8 checkpoint of Qwen3-32B, which stores 8-bit floating-point weights, the official BF16 checkpoint of Qwen2.5-72B, and, for Llama~4, the FP8 checkpoint published by Red Hat AI.
Throughout, Qwen2.5-72B refers to the instruction-tuned release, Qwen2.5-72B-Instruct.
The second run applies \protname's quantization to the same weights.
We compare the two runs by their perplexity ratio, which measures whether the quantization makes the model worse at predicting the text.

\parhead{Computing the perplexity}
Perplexity is computed from how well the model predicts each token of the text given the tokens before it, its context.
The result depends on how much context each token receives, so for each model we follow the convention of the baseline it is compared with.
For Qwen3-32B and Llama~4 we follow llama.cpp.
The text is cut into separate chunks of 512 tokens, and each chunk is given to the model as a separate input, so that the model sees nothing of the other chunks.
The perplexity is computed from the predictions for the second half of each chunk with the first half as context.
For Qwen2.5-72B we follow MOSAIC.
A window of 2,048 tokens slides over the text and advances by 512 tokens at a time.
The perplexity is computed from the predictions for the last 512 tokens of each window, with the preceding 1,536 tokens as context.

\parhead{Results}
\cref{tab:eval-accuracy} reports the results.
For the two Qwen models, the perplexity with \protname's quantization is within 0.04\% of the reference.
For Llama~4, it is within 0.2\% of the reference.
\protname's quantization therefore costs practically no accuracy.

\begin{table}[t]
\centering
\small
\setlength{\tabcolsep}{4pt}
\begin{tabular}{lrrr}
\toprule
Model & Reference & \protname & Ratio \\
\midrule
Qwen3-32B & 7.9435 & 7.9464 & 1.00036 \\
Qwen2.5-72B & 3.6744 & 3.6748 & 1.00009 \\
Llama~4 & 5.1508 & 5.1414 & 0.99817 \\
\bottomrule
\end{tabular}
\caption{Perplexity on the WikiText-2 test set with the released weights (reference) and with \protname's quantization.}\label{tab:eval-accuracy}
\end{table}

\section{Conclusions and Future Directions}
\label{sec:conclusion}
We develop \protname, a protocol for private and verifiable LLM inference that delegates the model's linear operations using a new matrix--vector multiplication primitive with transparent preprocessing and efficient client verification.
When only verification is required, our implementation serves models of up to 400B parameters to clients that cannot run them on their own: a laptop and a desktop with a 16\,GB GPU generate tokens $81$--$270\times$ faster on \texttt{Qwen3-32B} than with local inference, and they use \texttt{Llama~4}, a 400B-parameter model, at $7$--$15$ and $24$--$33$ tokens per second.
With input privacy, \protname generates tokens $10$--$12\times$ faster than the concurrent system MOSAIC in a data-center setting, with precomputed masks.
An important direction for future work is reducing the communication rounds for private inference.

\section*{Acknowledgements}
Wenhao Wang is supported in part by a research grant from IC3.

\bibliographystyle{plain}
\bibliography{refs}
\crefalias{section}{appendix}
\crefalias{subsection}{appendix}
\appendix
\section{Learning Parity with Noise}\label{app:lpn}
We begin by defining the dual-LPN assumption over a field $\F$ with respect to a distribution of code generators $\calG$, dimension $n$, block length $N>n$, and noise distribution $\calE$.

\begin{definition}[Dual-LPN]
    The $(\calE,\calG,\F)$-dual-LPN assumption for dimensions $n=n(\lambda)$ and $N=N(\lambda)$ states that
    \[
    (\bG,\bG\be)\approx_c (\bG,\bu)
    \]
    where $\bG\leftarrow \calG,\be\leftarrow \calE$, and $\bu\leftarrow \F^n$. Here, $\calG$ samples code generators $\bG\in\F^{n\times N}$ and $\calE$ samples noise vectors $\be\in\F^N$. Further, we say $\calG$ is LPN-friendly if this assumption is believed to hold under standard noise distributions $\calE$.
\end{definition}
Observe that $\bG\be$ is not a \textit{codeword}, but is instead a \textit{syndrome} of the dual code. As such, the search problem version of dual-LPN is also referred to as \textit{syndrome decoding}~\cite{expand-accumulate-codes}.

\parhead{Equivalence to (primal) LPN} 
The dual-LPN assumption is equivalent to the LPN assumption~\cite{lpn-assumption}, which informally states that 
\[
(\bH,\bH\bs+\be)\approx_c (\bH,\bu)
\]
where $\bH\in\F^{N\times n'}$ is a parity-check matrix for the code dual to $\bG$ and $\bs\leftarrow \F^{n'},\be\leftarrow \calE$, and $\bu\leftarrow \F^N$. Equivalence holds for $n'=N-n$ by solving for $\bG\in\F^{n\times N}$ such that $\bG\bH=0$ via row-reduction. It follows that
\[
\bG(\bH\bs+\be)=(\bG\bH)\bs+\bG\be=\bG\be
\]
If dual-LPN holds over a distribution of generators $\calG$ of a dual code, primal-LPN holds over the corresponding distribution of parity-check matrices $\calH$ for the primal code. See, e.g., ~\cite{efficient-pcgs,EMVP25} for further details.

\parhead{Selecting a noise distribution $\calE$.} For LPN-based constructions, exact noise is typically used~\cite{expand-accumulate-codes}, where $\calE=\calE_t^N$ is the uniform distribution over the set of $t$-sparse vectors in $\F^N$. However, alternative distributions such as Bernoulli noise or regular noise may be considered.

\parhead{Selecting code generators $\calG$.} A natural question arises of which $\calG$ are LPN-friendly. The standard LPN assumption instantiates $\calG$ as the uniform distribution over matrices in $\F^{n\times N}$. More recently, two separate lines of work in code-based encryption~\cite{efficient-encryption-quasicyclic,Aragon2021BIKEBF} and pseudorandom correlation generators~\cite{compressing-vector-ole,efficient-pcgs,expand-accumulate-codes,expand-convolute-codes} have sought to instantiate LPN from families of codes which enable linear or quasilinear encoding time. Such code families include quasi-cyclic, Toeplitz, LDPC~\cite{ldpc-codes}, expand-accumulate (EA)~\cite{expand-accumulate-codes}, expand-convolute (EC)~\cite{expand-convolute-codes}, and repeat-accumulate-accumulate (RAA)~\cite{blaze-fast-snarks-from-RAA-codes}, and more~\cite{ring-lpn-codes,EMVP25}. In order to test the concrete security of various possibly LPN-friendly codes, \cite{expand-accumulate-codes} introduced the following unified framework.

\parhead{Security against linear tests.} Suppose $\calG$ is a distribution over code generators $\bG$ such that the code does not admit an efficient decoding algorithm. For instance, Reed-Solomon codes are highly algebraic and enable fast decoding via the Berlekamp-Massey algorithm~\cite{berlekamp-massey}, which trivially breaks the LPN assumption for these codes. Non-algebraic codes typically do not admit such decoding algorithms.

The key observation of~\cite{expand-accumulate-codes} is that all known attacks on LPN over such codes reduce to finding a vector $\bv\in\F^n$ (or equivalently a linear functional) such that LPN samples have some non-negligible bias with respect to $\bv$. This procedure is called a \textit{linear test} and is done by analyzing $\bv^\top \bb$ for $\bb=\bG\be\in\F^n$ and $\be\leftarrow \calE$. Linear tests encapsulate information set decoding (ISD) attacks~\cite{information-set-decoding}, the BKW algorithm~\cite{bkw}, Gaussian elimination, and all other existing practical distinguishing attacks~\cite{expand-accumulate-codes}.

Then, \cite{expand-accumulate-codes} shows that the codes which are resistant to linear tests are precisely those which have high minimum distance, as this ensures that the inner product $\bv^\top \bb=(\bv^\top \bG)\be$ has $\bv^\top \bG$ of large Hamming weight and sufficiently high probability of colliding with sparse $\be$. We formally state these ideas below.

\begin{definition}[Bias of a distribution] For a nonzero vector $\bv\in\F^n$ and a distribution $\calD$ over $\F^n$, the bias of $\calD$ with respect to $\bv$ is defined as
\[
\bias_\bv(\calD)\coloneqq\max_{c\in\F}{\left|\Pr_{\bx\leftarrow \calD}\left[\bv^\top \bx=c\right]-\frac{1}{|\F|}\right|}.
\]
The bias of $\calD$ is defined as
\[
\bias(\calD)=\max_{\bv\neq 0}\bias_\bv(\calD).
\]
\end{definition}
For a noise distribution $\calE\subset \F^N$ and a code generator $\bG\in\F^{n\times N}$, we let $\calE_{\bG}\subset \F^n$ denote the distribution of $\bG\be$ where $\be\leftarrow \calE$, as in~\cite{abbaszadeh2025single}. Then, \textit{linear testing} corresponds to the search problem where a PPT adversary must output some vector $\bv\in\F^n$ such that $\bias_{\bv}(\calE_\bG)>\negl(\lambda)$. As shown in~\cite[Lemma 2.6]{expand-accumulate-codes}, if the minimum distance of the code generated by $\bG$ is $\ge d$, then
\[
\bias(\calE_\bG)
=
\max_{\substack{\bu\in\F^N\\\bu=\bv^\top \bG \text{ for some }\bv\neq 0}}\bias_\bu(\calE)
\leq 
\max_{\substack{\bu\in\F^N\\\wt(\bu)\ge d}}\bias_\bu (\calE).
\]
For $\calE=\calE_t^N$ the uniform distribution over $t$-sparse vectors in $\F^N$, the bias for any $\bu\in\F^N$ with Hamming weight at least $d$ can be bounded by 
\[
\bias_\bu(\calE_t^N)\leq \left(1-\frac{t}{N}\right)^d < e^{-td/N}=e^{-\delta t} .
\]
where $\delta=d/N$ is the relative distance of the given code~\cite{expand-accumulate-codes}. Since the search problem requires finding a suitable vector $\bu\in\F^N$, we can attain $\lambda-\log_2N$ bits of security as in~\cite{expand-accumulate-codes,abbaszadeh2025single} by setting $
e^{-\delta t}\leq 2^{-\lambda +\log_2N}$, or equivalently,
\[
t\geq \frac{(\ln 2)(\lambda-\log_2N)}{\delta}
\]
where $\delta$ depends on the field $\F$ and dimensions $n$ and $N$. For an asymptotically good code with $\delta=\Theta(1)$, this justifies the parameter selection of $t=O_\lambda(1)$.

The linear-test framework does not account for when the noise distribution is highly structured and an adversary is given sufficiently many samples, which enables algebraic attacks such as Arora-Ge~\cite{arora-ge}. We ensure that noise is generated via the \textit{exact distribution} $\calE_t^N$ to mitigate this risk~\footnote{\cite{expand-accumulate-codes} also remarks that when instantiating the linear-test framework over a ring, one must be wary of projection attacks onto ideals. These do not concern our construction since fields do not contain non-trivial ideals.}.

\section{Repeat-Accumulate-Accumulate Codes}\label{app:RAA-codes}
We instantiate the dual-LPN assumption with $\calG$ being a generator for Repeat-Accumulate-Accumulate (RAA) linear codes. We provide a summary of these codes and previous work supporting the conjecture that RAA codes are LPN-friendly.

For any field $\F$, dimension $n$, block length $N$ divisible by $n$, and constant $r=N/n$, an RAA code with generator $\GRAA\in\F^{n\times N}$ is constructed via the composition
\[
\GRAA=\bR\cdot \bPi_1\cdot \bV_1\cdot \bA\cdot \bPi_2\cdot \bV_2\cdot \bA.
\]
Each of the component matrices is sampled and defined as follows:
\begin{itemize}
    \item Repetition matrix $\bR\in\F^{n\times N}$: repeats each entry of the input vector $r$ times, given by \[
    R_{i,j}=\begin{cases}1& \text{if } (i-1)r< j\le ir\\0&\text{otherwise}\end{cases}.
    \]
    \item Permutation matrix $\bPi\in\F^{N\times N}$: samples and performs a uniformly random permutation $\pi:[N]\to[N]$ to the input vector, given by
    \[
    \Pi_{i,j}=\begin{cases}
        1&\text{if }\pi(j)=i\\
        0&\text{otherwise}
    \end{cases}.
    \]
    \item Randomization matrix $\bV\in\F^{N\times N}:$ samples a uniformly random nonzero vector $\bv\leftarrow(\F^\ast)^N$ along the diagonal of a matrix
    \[
    V_{i,j}=\begin{cases}
        v_i&\text{if }i=j\\
        0&\text{otherwise}
    \end{cases}.
    \]
    \item Accumulation matrix $\bA\in\F^{N\times N}:$ computes the prefix-sum of the input vector by applying the upper triangular matrix
    \[
    A_{i,j}=\begin{cases}
        1&\text{if }i\leq j\\
        0&\text{otherwise}
    \end{cases}.
    \]
\end{itemize}

\parhead{Linear-time encoding.} We observe that the encoding map $\F^n\rightarrow \F^N$ by $\bw\mapsto \GRAA^\top \bw$ can be computed in $O(N)$ field operations. By the transposition principle~\cite{1956-bordewijk}, the syndrome map $\F^N\rightarrow \F^n$ by $\bv\mapsto \GRAA \bv$ can similarly be computed in $O(N)$ field operations. We show that for each component matrix of $\GRAA$, matrix--vector multiplication can be computed in time $O(N)$.
\begin{itemize}
    \item Repetition matrix $\bw\mapsto \bR^\top \bw:$ requires $N$ copies of field elements---no field operations.
    \item Permutation matrix $\bw\mapsto \bPi^\top \bw:$ requires $N$ copies of field elements---no field operations.
    \item Randomization matrix $\bw\mapsto \bV^\top \bw$: requires $N$ field multiplications.
    \item Accumulation matrix $\bw\mapsto \bA^\top \bw:$ requires $N$ field additions to compute the prefix-sum.
\end{itemize}
Composing the above operations yields a circuit which computes encodings $\bw\mapsto \GRAA^\top \bw$ in at most $4N$ field operations. The concrete count for syndromes $\bv\mapsto \GRAA \bv$ sits similarly at $5N$ field operations.

\parhead{LPN friendliness and asymptotic goodness.} RAA codes are non-algebraic due to the accumulation and permutation matrices. As such, it is conjectured that no efficient decoding algorithm exists, and that RAA codes are LPN-friendly~\cite{EMVP25,abbaszadeh2025single}.

To motivate this conjecture, many previous papers have conducted analyses on the minimum distance of RAA codes, from which security against linear tests directly follows. Early analyses focused on asymptotic bounds for the minimum distance of RAA codes and limited their instantiation to binary fields~\cite{4729760,DBLP:journals/corr/abs-0810-3422,10.1109/TIT.2009.2030459}. Later works~\cite{5723043,1273662,5743575,cryptoeprint:2026/1075} showed similarly strong minimum-distance properties and asymptotic goodness for RAA codes over arbitrary prime fields.

Since asymptotic analysis is insufficient for instantiating cryptographic protocols with provable guarantees, further work has provided concrete bounds on the minimum distance of RAA codes over various parameters. Binary fields are analyzed in the Blaze SNARK~\cite{blaze-fast-snarks-from-RAA-codes} and proven to achieve constant relative distance, while \cite{abbaszadeh2025single} analyzes prime fields for $n\in[2^4,2^{10}]$ using combinatorial techniques and extrapolates to larger $n$. More recently,~\cite{distanceRAAcodes} gives a rigorous prime-field analysis proving a constant relative distance which is larger than in the binary case.

\parhead{Failure probabilities.} Crucially, when considering code generators $\GRAA$ sampled from a distribution, guarantees of good minimum distance are probabilistic rather than exact. That is, for a specified dimension $n$, block length $N$, rate $r=N/n$, field $\F$, and relative distance $\delta$, the code has a \textit{failure probability} that a randomly sampled generator matrix $\GRAA$ from this distribution will not have relative distance at least $\delta$. We denote this failure probability by $\pfail = \pfail(\calG, N, \delta)$. The focus of prior work is to demonstrate that the failure probability can be made arbitrarily small under practical parameters. 

In particular, \cite{distanceRAAcodes} shows that for $\delta=0.5$, the failure probability is $\widetilde{O}(N^{2-r})+O(\frac{1}{q-1})$, and \cite{cryptoeprint:2026/1075} sharpens this to $\widetilde{O}(N^{1-r})$ with no additive $1/(q-1)$ term, whenever $q-1=\Omega(N)$. Thus, setting $r=N/n\geq 3$ ensures that the failure probability decreases polynomially with $N$. Furthermore, a framework introduced in Blaze~\cite{blaze-fast-snarks-from-RAA-codes} enables publicly \textit{testing} the distance of an RAA code in polynomial time $O(N^w)$ for any integer $w\in\mathbb{N}$ in order to reduce the failure probability to $\widetilde{O}(N^{(w+1)(2-r)})$~\cite{distanceRAAcodes}.

\section{Private Matrix Extension}\label{app:private-matrix}
$\PipvMVMD$ can be extended to support a client who holds a \textit{private} matrix $\bM\in\F^{m\times n}$ in addition to a private vector $\bx\in\F^n$, and wishes to delegate $\bM\bx$ without revealing either operand to the compute server. In private \ac{llm} inference, this corresponds to outsourcing linear-layer computation for a proprietary model owned by the client.

We next formalize its syntax and security.

\begin{definition}[Private-Matrix pvMVMD]\label{def:private-matrix-pvMVMD}
A private-matrix \ac{pvMVMD} protocol $\Pi$ is a tuple of PPT algorithms $\textsc{(Preprocess,Encrypt,Verify,Decrypt)}$ defined by:
\begin{itemize}[leftmargin=*]
    \item $\textsc{Preprocess}(1^\lambda,1^\kappa,\bM)\rightarrow (\pp,\widehat{\bM},\st_\bM)$. On input security parameter $\lambda$, statistical parameter $\kappa$, and matrix $\bM\in\F^{m\times n}$, outputs public parameters $\pp$, masked matrix $\widehat{\bM}$, and private matrix state $\st_\bM$.
    \item $\textsc{Encrypt}(\pp,\bx)\rightarrow (\widehat{\bx},\st_\bx)$. On input public parameters $\pp$ and vector $\bx\in\F^n$, outputs ciphertext $\widehat{\bx}$ and private vector state $\st_\bx$.
    \item $\textsc{Verify}(\pp, \widehat{\bx},\widehat{\by})\rightarrow \{0,1\}$. On input public parameters $\pp$, ciphertext $\widehat{\bx}\in\F^n$, and server output $\widehat{\by}\in\F^m$, outputs $1$ (accept) or $0$ (reject).
    \item $\textsc{Decrypt}(\pp,\st_\bM,\st_\bx,\widehat{\by})\rightarrow \by$. On input public parameters $\pp$, private states $\st_\bM,\st_\bx$, and ciphertext $\widehat{\by}\in\F^m$, outputs plaintext vector $\by \in \F^m$.
\end{itemize}
The protocol satisfies:
\begin{itemize}[leftmargin=*]
    \item \textbf{Completeness}: For any $\bM\in\F^{m\times n}$ and $\bx \in \F^{n}$,
\[
\resizebox{0.93\columnwidth}{!}{$
\Pr\left[\by = \bM \bx \;\land\;  b=1
\;\middle|\;\begin{aligned}
    &(\pp,\widehat{\bM},\st_\bM)\leftarrow\textsc{Preprocess}(1^\lambda,1^\kappa,\bM)\\
    &(\widehat{\bx},\st_\bx)\leftarrow\textsc{Encrypt}(\pp,\bx)\\
    &\widehat{\by}=\widehat{\bM}\widehat{\bx} \\
    & b \leftarrow  \textsc{Verify}(\pp,\widehat{\bx},\widehat{\by}) \\
    &\by\leftarrow \textsc{Decrypt}(\pp,\st_\bM,\st_\bx,\widehat{\by}) 
\end{aligned}\right]
=1.
$}
\]

\item \textbf{Soundness}: There exists a negligible function $\negl(\cdot)$ such that for any $\bM\in\F^{m\times n}$ and $\bx \in \F^n$,
\[
\resizebox{0.93\columnwidth}{!}{$
\Pr\left[
\by\ne \bM\bx\;\land\; b=1
\;\middle|\;\begin{aligned}
    &(\pp,\widehat{\bM},\st_\bM)\leftarrow\textsc{Preprocess}(1^\lambda,1^\kappa,\bM)\\
    &(\widehat{\bx},\st_\bx)\leftarrow\textsc{Encrypt}(\pp,\bx)\\
    & \widehat{\by}\leftarrow\adv(\pp,\widehat{\bM},\widehat{\bx}) \\
    & b \leftarrow \textsc{Verify}(\pp,\widehat{\bx},\widehat{\by})\\
    &\by\leftarrow \textsc{Decrypt}(\pp,\st_\bM,\st_\bx,\widehat{\by})
\end{aligned}\right]
<\negl(\kappa).
$}
\]

\item \textbf{Client privacy}: There exists a PPT simulator $\mathsf{Sim}$ such that for any $\bM\in\F^{m\times n}$ and $\bx\in\F^n$,
\[
\mathsf{Sim}(1^\lambda, 1^\kappa, 1^m, 1^n) \approx_c \mathsf{View}_\textnormal{Server}^{\Pi}(\bM, \bx),
\]
where $\mathsf{View}_\textnormal{Server}^{\Pi}(\bM, \bx)\coloneqq(\pp, \widehat{\bM}, \widehat{\bx})$ is generated by $(\pp,\widehat{\bM},\st_\bM)\leftarrow\textsc{Preprocess}(1^\lambda,1^\kappa,\bM)$ and $(\widehat{\bx},\st_\bx)\leftarrow\textsc{Encrypt}(\pp,\bx)$.
\end{itemize}
\end{definition}

Our construction leverages the notion of a \textit{trapdoored matrix}~\cite{vaikuntanathan2026improving,braverman2025practical}, which is a pseudorandom matrix sampled alongside a trapdoor that enables fast \acp{mvm} by that matrix.

\begin{definition}[Trapdoored matrix~\cite{vaikuntanathan2026improving,braverman2025practical}]
A trapdoored matrix scheme $\textsc{TDM}$ is a pair of PPT algorithms $(\textsc{TDM.Gen},\textsc{TDM.Eval})$ with the following syntax:
\begin{itemize}[leftmargin=*]
    \item $\textsc{TDM.Gen}(1^\lambda,1^m,1^n) \to (\bM', \td_{\bM'})$ outputs a matrix $\bM'\in\F^{m\times n}$ and a trapdoor $\td_{\bM'}$.
    \item $\textsc{TDM.Eval}(\td_{\bM'},\bx)$ takes trapdoor $\td_{\bM'}$ and vector $\bx\in\F^n$, and outputs a vector in $\F^m$.
\end{itemize}
$\textsc{TDM}$ satisfies the following properties:
\begin{itemize}[leftmargin=*]
    \item \textbf{Correctness}: For every $(\bM',\td_{\bM'})$ output by $\textsc{TDM.Gen}(1^\lambda,1^m,1^n)$ and every $\bx\in\F^n$,
    \[
        \textsc{TDM.Eval}(\td_{\bM'},\bx) = \bM'\bx.
    \]
    \item \textbf{Pseudorandomness}: The distribution of $\bM'$ is computationally indistinguishable from the uniform distribution over $\F^{m\times n}$ to any PPT adversary that does not know $\td_{\bM'}$.
\end{itemize}
\end{definition}

Braverman and Newman~\cite{braverman2025practical}
and Vaikuntanathan and Zamir~\cite{vaikuntanathan2026improving} both construct TDMs using recursion, where $\textsc{TDM.Eval}$ runs in quasilinear time $T_{\textsc{TDM}}(m,n)=\widetilde{O}(m+n)$. Using dual-LPN-friendly linear-time codes, EMVP~\cite{EMVP25} improves the runtime to linear $T_{\textsc{TDM}}(m,n) = O(m+n)$.

In our construction, the client masks $\bM$ by sampling $(\bM',\td_{\bM'})\leftarrow\textsc{TDM.Gen}$ and uploading $\widehat{\bM}\coloneqq\bM+\bM'$ to the server, storing $\td_{\bM'}$ locally. The preprocessing terms $\bP\coloneqq\widehat{\bM}\privG$ and $\bQ\coloneqq\verG^\top\widehat{\bM}$ are generated as part of $\pp$. For each query, the client masks $\bx$ as $\widehat{\bx}\coloneqq\bx+\privG\prive$, and the server computes $\widehat{\by}\coloneqq\widehat{\bM}\widehat{\bx}=(\bM+\bM')(\bx+\bx')$. The client checks the claim using $\textsc{Verify}$ and decrypts $\by\coloneqq\widehat{\by}-\bP\prive-\textsc{TDM.Eval}(\td_{\bM'},\bx)$. \Cref{prot:private-matrix-pvMVMD} presents the full protocol.

\begin{protocol}[!t]
\begin{mdframed}
\small
\setlist[itemize]{nosep}
\begin{center}
\colorbox{gray!20}{$\PipvMVMD^{\mathsf{privM}}$}
\end{center}
\underline{\textbf{Input parameters}}:
\vspace{0.1em}
\begin{itemize}[leftmargin=*]
    \item Trapdoored matrix scheme $\textsc{TDM} = (\textsc{TDM.Gen}, \textsc{TDM.Eval})$ for matrices in $\F^{m\times n}$.
    \item Code distributions $\privcalG$ and $\vercalG$, sparsity parameters $\privt, \verit$, repetition parameter $\verl$, and block lengths $\privN, \verN$ as in \cref{prot:pvMVMD}.
    \item Field oracle $\F$.
\end{itemize}

\vspace{0.5em}
\underline{\textbf{\textsc{Preprocess}}$(1^\lambda,1^\kappa,\bM)\to(\pp,\widehat{\bM},\st_{\bM})$}:
\vspace{0.1em}
\begin{itemize}[leftmargin=*]
    \item Sample $(\bM',\td_{\bM'}) \leftarrow \textsc{TDM.Gen}(1^\lambda,1^m,1^n)$
    \item Compute $\widehat{\bM} \coloneqq \bM+\bM' \in\F^{m\times n}$.
    \item Sample $\privG\leftarrow\privcalG$ and $\verG\leftarrow\vercalG$.
    \item Compute $\bP \coloneqq \widehat{\bM}\privG \in\F^{m\times\privN}$.
    \item Compute $\bQ \coloneqq \verG^\top\widehat{\bM} \in\F^{\verN\times n}$.
    \item Output $\pp \coloneqq (\privG, \verG, \bP, \bQ, \privt, \verit, \verl)$, $\widehat{\bM}$, and $\st_{\bM} \coloneqq \td_{\bM'}$.
\end{itemize}

\vspace{0.5em}
\underline{\textbf{\textsc{Encrypt}}$(\pp,\bx)\to(\widehat{\bx},\st_{\bx})$}:
\vspace{0.1em}
\begin{itemize}[leftmargin=*]
    \item Sample $\prive\leftarrow\calE_{\privt}^{\privN}$ and compute $\bx' \coloneqq \privG\prive \in\F^n$.
    \item Output $\widehat{\bx} \coloneqq \bx+\bx'$ and $\st_{\bx} \coloneqq (\prive,\bx)$.
\end{itemize}

\vspace{0.5em}
\underline{\textbf{\textsc{Verify}}$(\pp,\widehat{\bx},\widehat{\by})\to\{0,1\}$}:
\vspace{0.1em}
\begin{itemize}[leftmargin=*]
    \item Compute $\verG^\top \widehat{\by}\in\F^{\verN}$.
    \item Repeat $\verl$ times:
    \begin{itemize}[leftmargin=*]
        \item Sample $\vere\leftarrow\calE_{\verit}^{\verN}$.
        \item Compute $v_L \coloneqq (\vere^\top\bQ)\widehat{\bx}$ and $v_R \coloneqq \vere^\top(\verG^\top\widehat{\by})$.
        \item If $v_L\neq v_R$, output 0 (reject) and halt.
    \end{itemize}
    \item If all $\verl$ checks pass, output 1 (accept).
\end{itemize}

\vspace{0.5em}
\underline{\textbf{\textsc{Decrypt}}$(\pp,\st_{\bM},\st_{\bx},\widehat{\by})\to\by$}:
\vspace{0.1em}
\begin{itemize}[leftmargin=*]
    \item Parse $\st_{\bM} = \td_{\bM'}$ and $\st_{\bx} = (\prive,\bx)$.
    \item Compute $\by_1 \coloneqq \bP\prive$ and $\by_2 \coloneqq \textsc{TDM.Eval}(\td_{\bM'},\bx)$.
    \item Output $\by \coloneqq \widehat{\by} - \by_1 - \by_2 \in\F^m$.
\end{itemize}
\end{mdframed}
\caption{Protocol $\PipvMVMD^{\mathsf{privM}}$ for private-matrix and private-vector delegation.}
\label{prot:private-matrix-pvMVMD}
\end{protocol}

\begin{theorem}
\label{thm:private-matrix-pvMVMD-secure}
Suppose that $\textsc{TDM}$ satisfies correctness and pseudorandomness, $\privcalG$ is a distribution over codes with relative distance at least $\privdelta$ satisfying the $(\calE_{\privt}^{\privN},\privcalG,\F)$-dual-LPN assumption for polynomially many independent samples, and the code generated by $\verG$ has relative distance $\verdelta$.
Then $\PipvMVMD^{\mathsf{privM}}$ satisfies completeness, soundness, and client privacy as defined in \cref{def:private-matrix-pvMVMD} when instantiated with $|\F|\geq 4$, $\verl = \left\lceil \frac{\kappa}{\log_2(|\F|-1)-1} \right\rceil$, and $\verit = \left\lceil \frac{\max\left(\log_2(|\F|-1),\kappa+1\right)}{-\log_2(1-\verdelta)} \right\rceil$.
\end{theorem}
\begin{proof}
For any vector $\widehat{\by}$, the decrypted output is $\by \coloneqq \widehat{\by} - \bP\prive - \textsc{TDM.Eval}(\td_{\bM'},\bx)$. Using $\bP = \widehat{\bM}\privG$, $\textsc{TDM.Eval}(\td_{\bM'},\bx) = \bM'\bx$, and $\widehat{\bM}=\bM+\bM'$, we derive
\begin{align*}
    \by - \bM\bx
    &= \widehat{\by} - \widehat{\bM}\privG\prive - \bM'\bx - \bM\bx \\
    &= \widehat{\by} - \widehat{\bM}\privG\prive - (\bM+\bM')\bx \\
    &= \widehat{\by} - \widehat{\bM}(\bx + \privG\prive) \\
    &= \widehat{\by} - \widehat{\bM}\widehat{\bx}.
\end{align*}

\noindent\textbf{Completeness.} If the server honestly evaluates $\widehat{\by} = \widehat{\bM}\widehat{\bx}$, then the identity above yields $\by=\bM\bx$. Further, \textsc{Verify} accepts with probability 1 by \cref{thm:vMVMD-secure}.

\noindent\textbf{Soundness.} If the server returns an output with $\by \neq \bM\bx$, the identity above implies $\widehat{\by} \neq \widehat{\bM}\widehat{\bx}$. By \cref{thm:vMVMD-secure}, \textsc{Verify} rejects with probability at least $1 - 2^{-\kappa}$ for these parameters.

\noindent\textbf{Client privacy.} We prove client privacy via two hybrid steps. In the first hybrid, we replace the trapdoored matrix $\bM'$ with a uniformly random matrix $\bU_{\bM}\leftarrow\F^{m\times n}$, which is computationally indistinguishable by the pseudorandomness of $\textsc{TDM}$. After this replacement, $\widehat{\bM}=\bM+\bU_{\bM}$ is uniformly distributed over $\F^{m\times n}$ and independent of $\bM$. In the second hybrid, we replace $\privG\prive$ with an independently sampled uniform vector $\bu\leftarrow\F^n$, which is computationally indistinguishable under the $(\calE_{\privt}^{\privN},\privcalG,\F)$-dual-LPN assumption. After this replacement, $\widehat{\bx}=\bx+\bu$ is uniformly distributed over $\F^n$ and independent of $\bx$. The resulting distribution is simulated by a PPT simulator $\mathsf{Sim}(1^\lambda, 1^\kappa, 1^m, 1^n)$ that samples uniform $\widetilde{\bM}\leftarrow\F^{m\times n}$ and $\widetilde{\bx}\leftarrow\F^n$, computes $\widetilde{\pp}$ from $\widetilde{\bM}$, and outputs $(\widetilde{\pp}, \widetilde{\bM}, \widetilde{\bx}) \approx_c \mathsf{View}_\textnormal{Server}^{\Pi}(\bM,\bx)$. Computational indistinguishability follows from these two hybrids.
\end{proof}

\parhead{Efficiency.}
The online client computes $\textsc{TDM.Eval}$ in some time $T_{\textsc{TDM}}(m,n)$ and all other $\textsc{Encrypt}$, $\textsc{Verify}$, and $\textsc{Decrypt}$ steps in time $O_{\lambda,\kappa}(m+n)$ by~\cref{subsec:delegation:efficiency}. Choosing a linear evaluation time $\textsc{TDM}$ (e.g., in~\cite{EMVP25}) provides asymptotically optimal online time $O_{\lambda,\kappa}(m+n)$. The server performs exactly one field \ac{mvm}, and preprocessing requires $O(mn)$ field operations to compute $\bP$ and $\bQ$.

\parhead{Drawback on preprocessing}
In the public-matrix setting \ac{pvMVMD}, preprocessing is \emph{transparent} and can be computed once by the server and published globally for all clients. In contrast, for private-matrix delegation, each client must sample their own independent trapdoor $\td_{\bM'}$ and mask the matrix $\widehat{\bM}$ privately, requiring $O(mn)$ \emph{private} preprocessing time. This construction is therefore most attractive when the preprocessing cost can be amortized across many queries.

\section{Soundness Composition}\label{app:soundness-composition}

We show that any private delegation protocol can be made sound by composing with a sound verifiable delegation protocol. This is the same technique we use in constructing $\PipvMVMD$ from $\PivMVMD$ in \cref{prot:pvMVMD}, and this demonstrates that a lightweight verification layer can be applied to a private delegation protocol.

Let $\PivMVMD$ be a complete and sound \ac{vMVMD} protocol, and let $\PipvMVMD$ be a complete and private but possibly unsound private delegation protocol. We informally define the protocol $\PipvMVMD'$ as follows:
\begin{itemize}[leftmargin=*]
    \item \textsc{Preprocess}: run $\PivMVMD.\textsc{Preprocess}$ and $\PipvMVMD.\textsc{Preprocess}.$
    \item \textsc{Encrypt}: run $\PipvMVMD.\textsc{Encrypt}.$
    \item \textsc{Verify}: run $\PivMVMD.\textsc{Verify}.$
    \item \textsc{Decrypt}: run $\PipvMVMD.\textsc{Decrypt}.$
\end{itemize}

Completeness and client privacy of $\PipvMVMD'$ follow from the corresponding properties of $\PivMVMD$ and $\PipvMVMD$.

For soundness, suppose the adversary returns $\widehat{\by}$ such that $\by\neq \bM\bx$, where $\by=\PipvMVMD'.\textsc{Decrypt}(\pp,\st,\widehat{\by})$. By the completeness of $\PipvMVMD$, if $\widehat{\by}=\bM\widehat{\bx}$ then $\by = \bM\bx$; taking the contrapositive, if $\by\neq \bM\bx$, it must hold that $\widehat{\by}\neq \bM\widehat{\bx}$. The $\textsc{Verify}$ algorithm of $\PivMVMD$ is applied to the claim
\[
\widehat{\by}\stackrel{?}{=}\bM\widehat{\bx}.
\]
By soundness of $\PivMVMD$, $\textsc{Verify}$ detects this inconsistency except with negligible probability.

\section{Attack on Slalom at the Carnival}\label{app:satc}
Slalom at the Carnival (S@C) by Bruhns et al.~\cite{satc} proposes an addition to Slalom~\cite{slalom} which reduces the client online running time required to encrypt and decrypt vectors in a \ac{pvMVMD} protocol. They claim to achieve client privacy; however, we show a flaw in their approach which breaks client privacy and enables a semi-honest server to uncover projections of the client's input vectors.

\parhead{S@C protocol} Fix a matrix $\bM\in\F^{m\times n}$ and a parameter $\ell \ll n$ which is dependent on the security parameter. In the S@C preprocessing, the server samples $\ell$ vectors $\br_1,\dots,\br_\ell\leftarrow \F^n$ and computes the corresponding projections $\bu_1,\dots,\bu_\ell\in\F^m$, where $\bu_i=\bM\br_i$. These projections are sent to the client. When the client wishes to delegate an \ac{mvm} for a vector $\bx$, they first sample $\bk\leftarrow \{0,1\}^\ell$, compute the masking term
\[
\bx':=\sum_{i=1}^\ell k_i\br_i\in\F^n,
\]
and encrypt $\widehat{\bx}:=\bx+\bx'$. This is sent to the server, which returns $\widehat{\by}=\bM\widehat{\bx}$. The client decrypts by subtracting the projections:
\[
\by:=\widehat{\by}-\sum_{i=1}^\ell k_i\bu_i.
\]
\parhead{The attack}
The S@C protocol is efficient only when $\ell < n$. However, in this parameter range, the server can use Gaussian elimination to solve for a projection matrix $\bPi\in\F^{(n-\ell)\times n}$ annihilating the vectors $\br_1,\dots,\br_\ell$ (i.e., $\bPi\br_i=0$ for each $i$).

For a client input $\widehat{\bx}$, the server applies $\bPi$ and recovers
\[
\bPi\widehat{\bx}=\bPi\left(\bx+\sum_{i=1}^\ell k_i\br_i\right)=\bPi\bx\in\F^{n-\ell}.
\]
This enables the server to distinguish $\widehat{\bx}$ from truly random and precludes the existence of a privacy simulator that works for all $\bx$.

\end{document}